\documentclass[12pt]{article}

\usepackage[T1]{fontenc}
\usepackage[margin = 3 cm]{geometry}
\usepackage{enumitem}
\usepackage{float}
\usepackage[backend=bibtex, style = numeric-comp]{biblatex}
\usepackage{url}
\usepackage[hidelinks]{hyperref}
\usepackage{xurl}
\hypersetup{breaklinks=true} %
\usepackage{xspace}
\usepackage{amsthm}
\usepackage{amsmath} %
\usepackage{amssymb} %
\usepackage{subcaption}
\usepackage[textwidth=2cm, disable]{todonotes}
\usepackage{amstext,amsfonts,bm}
\usepackage{thm-restate}
\usepackage[parfill]{parskip}
\usepackage{authblk}
\usepackage{graphicx}
\AtEveryBibitem{\clearfield{doi}}
\usepackage[linesnumbered, ruled, noend]{algorithm2e}
\usepackage{algorithmicx}
\usepackage{booktabs}

\makeatletter
\patchcmd{\algocf@makecaption@ruled}{\hsize}{\textwidth}{}{} %
\patchcmd{\@algocf@start}{-1.5em}{0em}{}{} %
\makeatother
\SetAlCapHSkip{0pt} %
\IncMargin{0.35cm}

\RequirePackage[T1]{fontenc} \RequirePackage[tt=false, type1=true]{libertine} \RequirePackage[varqu]{zi4} \RequirePackage[libertine]{newtxmath}%

\newtheoremstyle{own}%
{6pt}%
{3pt}%
{\itshape}%
{}%
{\bfseries}%
{.}%
{.5em}%
{}%

\theoremstyle{own}

\newcommand{\lIfElse}[3]{\lIf{#1}{#2 \textbf{else}~#3}}

\newcommand{\ps}[1]{\todo{\footnotesize\textbf{PS:} #1}}

\newcommand{\pr}{\ensuremath{\mathbb P}}
\newcommand{\ex}{\ensuremath{\mathbb E}}
\newcommand{\bigO}{\smash{\ensuremath{O}}}
\newcommand{\tilO}{\smash{\ensuremath{\widetilde O}}}
\newcommand{\p}{\ensuremath{\!+\!}}
\newcommand{\m}{\ensuremath{\!-\!}}

\newcommand{\poly}[1]{\text{poly}(#1)}

\newtheorem{theorem}{Theorem}
\newtheorem{lemma}{Lemma}[section]
\newtheorem{corollary}[lemma]{Corollary}
\newtheorem{definition}[lemma]{Definition}

\newtheorem{remark}[lemma]{Remark}

\begin{document}
\title{Byzantine Fault Tolerant Protocols with Near-Constant Work per Node without Signatures\thanks{Philipp Schneider has been supported by a grant from Avalanche, Inc.\ to the University of Bern. We thank Christian Cachin for valuable discussions.}}

\author[1,2]{Philipp Schneider}
\affil[1]{University of Bern} 
\affil[2]{\href{philipp.schneider2@unibe.ch}{philipp.schneider2@unibe.ch}}
 
\date{}
\maketitle              %
\begin{abstract}
	Numerous distributed tasks have to be handled in a setting where a fraction of nodes behaves Byzantine, that is, deviates arbitrarily from the intended protocol. Resilient, deterministic protocols rely on the detection of majorities to avoid inconsistencies if there is a Byzantine minority, which requires individual nodes to handle a communication workload that is proportional to the size of the network --- an intolerable disadvantage in large networks.	
		
	Randomized protocols circumvent this by probing only small parts of the network, thus allowing for consistent decisions quickly and with a high level of confidence with communication that is near-constant in the network size. However, such protocols usually come with the drawback of limiting the fault tolerance of the protocol, for instance, by severely restricting the number or type of failures that the protocol can tolerate.
	
	We present randomized protocols to reliably aggregate and broadcast information, form consensus and compute common coins that tolerate a constant fraction of Byzantine failures, do not require cryptographic signatures and have a near-constant time and message complexity per node.	
	Our main technique is to compute a system of witness committees as a pre-computation step almost optimally. This pre-computation step allows to solve the aforementioned distributed tasks repeatedly and efficiently, but may have far reaching further applications, e.g., for sharding of distributed data structures.
\end{abstract}

\ps{summary of parameters probably as some table in abstract}

\newpage

\tableofcontents

\newpage

\section{Introduction}

Distributed algorithms such as reliable broadcast and consensus are routinely used in applications such as maintaining replicated state machines, in particular for distributed databases, cloud computing and blockchains. When designing such algorithms, resilience to node failures is essential to ensure safety against technical malfunctions and adversarial attacks to maintain consistent system operation and trust.

Algorithms for consensus and reliable broadcast are widely used but inherently challenging, since fundamental impossibility results under various synchrony assumptions of the network \cite{DBLP:journals/jacm/FischerLP85, DBLP:journals/jacm/PeaseSL80} have to be taken into account. Most deterministic algorithms ensure consistency via quorums that intersect in non-failing nodes but are usually linear in size, thus limiting scalability. Scalable randomized approaches either sacrifice resiliency \cite{DBLP:conf/sirocco/AmoresSesarCS24, DBLP:conf/spaa/RobinsonSS18} or rely on computationally expensive cryptography \cite{DBLP:conf/sosp/GiladHMVZ17, DBLP:conf/opodis/BezerraAKS024}.

The main motivation of this paper is to address distributed problems such as consensus, reliable broadcast and computing common coins with following three main properties. 
\begin{itemize}
	\item The algorithms achieve near-constant work per node, meaning that the communication load per node grows negligibly with network size.
	\item Resilience against a constant fraction of Byzantine nodes that can deviate arbitrarily from the protocol, in a comparatively strong adversarial model.
	\item Avoiding cryptographic primitives such as public-key cryptography, signatures, threshold schemes or verifiable random functions. 
\end{itemize}
This has the following advantages. First, it allows scalability to massive networks where individual nodes have relatively low bandwidth. For example, for data aggregation in vast sensor networks or industrial internet of things \cite{DBLP:journals/cii/BoyesHCW18}, or for coordinating micro-grids and energy storage systems in decentralized energy grids, where thousands or even millions of nodes with limited computational power and bandwidth must collaborate while ensuring resilience towards node failures and minimizing overhead from cryptographic operations.

Second, we admit a constant fraction of Byzantine nodes and at the same time respect impossibility results by restricting the adversary in reasonable ways. Third, our approach eliminates computational overhead from the typically complex cryptographic operations the necessity to maintain a secret (which can be stolen) and reliance on certain assumptions, such as a computationally bounded adversary and the existence of trapdoor one-way functions required for public-key cryptography (certainly untrue if $P = NP$, but not necessarily true if $P \neq NP$). It also makes the solution future-proof in the event that quantum computers become practical, which will invalidate certain cryptographic protocols.

Typical distributed systems that run in large-scale networks require solving fundamental distributed problems such as consensus and reliable broadcast repeatedly throughout their lifetime to ensure their consistency. This raises the natural question whether the amortized cost can be decreased by dividing algorithms into a \textit{pre-computation} step and more efficient \textit{execution step} to solve each individual problem instance. 
We answer this question affirmatively by computing a generic structure called a \textit{system of witness committees} in the pre-computation step that is resilient to a constant fraction of adversarial nodes. We then leverage this structure in the execution step to solve aforementioned problems with near-constant time and workload per node.\footnote{While witness committees have been used in the context of reliable broadcast and consensus before to minimize total message complexity \cite{DBLP:conf/sosp/GiladHMVZ17, DBLP:conf/opodis/BezerraAKS024}, these approaches typically use only a single witness committee, have up to linear workload per node and rely heavily on cryptographic methods to obtain them. By contrast, we aim to compute a large number of such committees almost optimally without such cryptographic methods.} 

The pre-computation step establishes a large number of small node sets, the witness committees, each containing a majority of non-failing nodes, and makes them known network-wide. 
We show that this approach allows for a degree of parallelization that significantly improves efficiency in terms of round complexity and per-node workload in the execution step. Systems of witness committees have potential applications beyond the problems discussed here, perhaps most prominently for sharding of blockchains.

The separation into a pre-computation and execution step also introduces a degree of modularity. Notably, any assumptions that we make about the adversary are confined to the pre-computation phase to obtain a system of witness committees, as the execution phase relies only on the existence of such witness committees. This also admits the system to be permissionless to a certain degree after the pre-computation phase, since nodes can be added without jeopardizing the honest majorities within the committees. Furthermore, we believe that this generic approach allows for easier adaptation to varying performance requirements or modeling assumptions in both the pre-computation and execution phases, facilitating further improvements in future research.

\subsection{Contributions}

We start with some basic definitions sufficient to state our contributions.
We consider a network of $n$ nodes which we denote with the set $V$ (standing for vertices, or validators). The set of nodes $V = V_h + V_b$ is partitioned into a set of \textit{honest} nodes $V_h$ that adhere to the protocol and \textit{Byzantine} nodes $V_b$ that may deviate arbitrarily from the protocol. Besides an upper bound $t \geq |V_b|$, this partition is unknown to honest nodes. We consider a worst case adversary controlling the Byzantine nodes $V_b$, which can be seen as the nodes $V_b$ colluding in order to undermine the protocol.

The network is assumed to have reliable, private point-to-point links between any pair of nodes, meaning that we assume messages are not modified when sent over the network and a receiver can always identify the sender. In this paper, the (communication) workload per node refers to the number of bits a node has to send during an algorithm.
Specifically, we assume that nodes have a limited bandwidth for sending information formalized with $\sigma$ denoting the number of bits each node may send per round. Nodes are capable to receive any number of bits per round, but our algorithms ensure that without traffic from Byzantine nodes the number of bits received per round is $\bigO(\sigma)$. 

This paper focuses on distributed algorithms that optimize the communication among nodes, therefore we assume that computations based on local data are completed instantaneously. We do not exploit this assumption as our local computations are relatively light-weight, in particular our algorithms do not require computation heavy cryptographic signatures. Furthermore, we focus on providing algorithms that are \textit{almost optimal}, i.e., optimal up to factors that are polynomial in $\log n$. For this reason we use the $\tilO(\cdot)$ notation that suppresses any $\log n$ terms to simplify this part of the introduction (references to the detailed results are provided). Specifically, the expression \textit{near-constant} can be equated to $\tilO(1)$.

\paragraph{Part 1: System of Witness Committees} 
In the first part of this article, we present our main technical contribution: an efficient method for computing a \textit{system of witness committees}. The concept of witness committees is crucial because it provides quorums of very small size, minimizing communication overhead in the Byzantine setting \cite{DBLP:conf/sosp/GiladHMVZ17, DBLP:conf/opodis/BezerraAKS024}. Here, we compute a linear number of such committees with specific properties as a pre-computation step. As we discuss in the second part, this approach enables efficient solutions to various distributed problems with low per-node workload, and it has further applications beyond those examples.

Informally, a witness committee for a node $u \in V$ is a near-constant-size subset of nodes associated with $u$. A complete system provides such committees for a constant fraction of nodes, ensuring that no node belongs to too many committees and that each committee has a majority of honest (non-Byzantine) members. Moreover, all honest nodes must roughly agree on these committees: for every valid committee, their local views share the same common core of honest nodes, differing only by a small number of Byzantine ones. Any node $u$ without a committee is considered \emph{invalid} and is recognized as such by all nodes. A formal definition appears in Definition~\ref{def:witness_committees}.

Unfortunately, we cannot hope to compute such a system of witness committees deterministically, even if the adversary’s capabilities are limited or the number of Byzantine nodes is small. Any such deterministic procedure would require foreknowledge of how Byzantine nodes are distributed across the network, which contradicts the Byzantine model. This is why randomization is crucial. A common idea, used in prior work (e.g., \cite{DBLP:journals/corr/abs-1906-08936, DBLP:conf/spaa/DoerrGMSS11, DBLP:conf/sosp/GiladHMVZ17, DBLP:conf/opodis/BezerraAKS024}), is that a relatively small random sample tends to have a Byzantine-to-honest ratio mirroring that of the entire network.

However, a simplistic approach---where a single node creates a random sample and shares it with everyone---conflicts with our goal of establishing a complete \textit{system} of witness committees that (A) ensures agreement among all honest nodes, (B) achieves asymptotically optimal time and message complexity, and (C) avoids heavy cryptographic machinery. Naive methods are easily undermined by an adversary controlling a constant fraction of the nodes, leading to disagreement about the composition of committees or committees with a Byzantine majority---an outcome that is typically mitigated by costly cryptographic signatures or all-to-all broadcasts.

Our solution constructs witness committees through multiple phases, each producing a more refined system that meets increasingly strict requirements. The reason for splitting the process into phases is that we cannot afford expensive all-to-all communication, where every node contacts a linear-sized quorum, just to construct a single committee. Furthermore, we cannot rely on information relayed by individual nodes in the absence of cryptographic signatures. Instead, each phase leverages the properties established in the previous one to operate more efficiently. However, we do require that the adversary controls at most a constant fraction of the nodes \textit{and} the total communication bandwidth.

This bandwidth restriction reflects a realistic limit on the adversary’s resources. Without it, an adversary controlling a large fraction (or majority) of total bandwidth could easily drown out any agreement, akin to a denial-of-service attack. In our algorithm, this limitation prevents the adversary from subverting too many committees in the first phase. Concretely, the adversary can only force a committee to become ``invalid'' or allow it to have an honest majority. Further, the adversary must ``expend'' substantial bandwidth to invalidate a single committee, consequently a constant fraction of committees remains valid.

The main technical effort takes place in two subsequent phases, which refine the preliminary
witness committees from the first phase. In the second phase, we establish agreement on which committees are valid or invalid by having nodes determine which committees are ``supported'' by 
sufficiently many others, while communicating with only a small subset of nodes. In the third phase, we propagate knowledge of committees deemed valid by a majority to those that do not yet know them. Overall, we prove the following theorem (a simplified version of Theorem~\ref{thm:set_up_committee}).

\begin{theorem}
	\label{thm:witness_committees}
	Given a Byzantine adversary that controls a constant fraction of nodes and overall bandwidth. There is a Monte Carlo algorithm to compute a system of witness committees (formally given in Definition \ref{def:witness_committees}) with $\Theta(n)$ witness committees of near-constant size in $\tilO(n)$ rounds and $\tilO(1)$ communication workload per node per round without using signatures.
\end{theorem}

Note that such a system of $\Omega(n)$ random witness committees has an overall Shannon entropy of at least $\Omega(n)$. 
Therefore, if we restrict communication per node to just $\tilO(1)$ this strictly requires $\widetilde \Omega(n)$ rounds and exchanged bits such that any given node learns each witness committee (even just approximately). Consequentially, Theorem \ref{thm:witness_committees} is almost at the theoretical limit (i.e., up to $\log n$ factors) of the required time and message complexity such that every node learns such a randomized system of witness committees.

In this work, we restrict ourselves to demonstrating that a system of witness committees with the desired characteristics can be computed almost optimally with a \textit{constant} fraction of adversarial nodes and controlled bandwidth. Our current assumption of the ratio of Byzantine control over the nodes and the total bandwidth is $1/24$ (i.e., $n/24$ and $\sigma n/24$, respectively). Note that this fraction is \textit{not} fully optimized, since we opt for a streamlined presentation over optimizing the Byzantine ratio, wherever such a choice had to be made. 

Essentially, whenever some ``slack'' is required to account for effects of randomization or Byzantine influence we choose the next ``nice'' fraction instead of a full optimization via an additional parameter describing this ``slack''. Since such situations occur multiple times throughout our analysis (with complex inter-dependencies) and because the effects are compounding in a multiplicative way, the resulting Byzantine ratio is significantly lower than what it could be if carefully optimized. We leave that open for future work.

\paragraph{Part 2: Applications for Systems of Witness Committees} In the second part of this work, we assume that a system of witness committees is given, and exploit it to solve a variety of distributed problems in the Byzantine setting with near-constant time and work per node without using signatures. First we consider the reliable broadcast problem which is defined as follows.

\begin{definition}[Reliable Broadcast]
	\label{def:reliable_broadcast}
	Let $s \in V$ be  a dedicated sender node that has some message. The reliable broadcast problem is solved if the following holds.
	\begin{itemize}
		\item \textbf{Integrity}: Every honest node $v \in V_h$ delivers at most one message.
		\item \textbf{Validity}: If $s$ is honest and has a message $M$, then every $v \in V_h$ eventually delivers $M$.
		\item \textbf{Agreement}: If any $v \in V_h$ delivers $M$, then all $v \in V_h$ deliver $M$.
	\end{itemize}	
\end{definition}

A system of witness committees as described above can be used for reliable communication without signatures as follows. First we implement a reliable broadcast between single nodes and committees and then between two committees. Then we construct a tree structure among the valid committees exploiting the common knowledge of all valid committees.
Relying on these subroutines, we can efficiently broadcast a message through the tree-structure among committees with near-constant work per node. 

The following theorem summarizes the more general statement of Theorem~\ref{thm:reliable_broadcast}. Crucially, all randomization and any assumptions about adversaries or synchrony are required only during the pre-computation phase for the system of witness committees. As a result, the broadcast algorithm itself is deterministic and designed for an asynchronous network.

\begin{theorem}
	\label{thm:reliable_broadcast_simplified}
	Given a system of witness committees there is a deterministic, reliable broadcast problem from Definition \ref{def:reliable_broadcast} with $\tilO(\delta)$ communication workload per node\footnote{Assuming near-constant message size. Else, workload per node is proportional to  message size.} that works even in the asynchronous setting and does not require signatures. In the partially synchronous model it takes $\bigO(\log_\delta n)$ rounds after global stabilization.
\end{theorem}

This result is parameterized by the degree $\delta$ of the broadcast tree among witness committees and reflects a fundamental trade-off between per-node workload (proportional to $\delta$) and round complexity (the broadcast tree's height). Choosing constant \smash{$\delta > 1$}, solves reliable broadcast with $\tilO(1)$ workload per node in $\bigO(\log n)$ rounds satisfying the characteristics advertised in the title. Other important inflection points are \smash{$\delta = n^{1/\log \log n}$}, which yields a sublinear workload per node and a round complexity of $\bigO(\log \log n)$  and $\delta = n^\varepsilon$ for arbitrarily small constant $\varepsilon > 0$, which even ensures a constant round complexity of \smash{$\bigO\bigl(\tfrac{1}{\varepsilon}\bigr)$}.

We can also apply our subroutines for reliable broadcast among single nodes and committees through the broadcast tree to solve the \emph{reliable aggregation} problem. Here, some subset of nodes contributes input values to a multivariate function whose partial computations can be performed in any order (e.g., minimum, sum, XOR; see Definition~\ref{def:reliable_aggregation}). Each node first commits its value to a dedicated witness committee via reliable broadcast, which can aggregate values correctly due to their honest majorities. This process is repeated by aggregating up the broadcast tree. The final result is broadcast from the root, down the tree, to all nodes. We require a synchronous network so that any node (including Byzantine ones) wishing to participate must commit their value within a specified time bound. We obtain the following theorem (a simplification of Theorem~\ref{thm:reliable_aggregation}).

\begin{theorem}
	\label{thm:reliable_aggregation_simplified}
	Given a system of witness committees as in Definition \ref{def:witness_committees}. We can solve \ref{alg:bc} the reliable aggregation problem  (see Definition \ref{def:reliable_aggregation}) in the synchronous setting without using signatures with $\tilO(\delta)$ total communication workload per node in $\bigO(\log_\delta n)$ rounds.
\end{theorem}

Our reliable aggregation routine can be used for computing a common coin beacon, which produces a shared coin in $\tilO(1)$ rounds and $\tilO(1)$ communication per node by taking the XOR of all random local coins that each node commits. As soon as one honest node's coin is included, the final coin is completely random. Besides synchrony in the first round, so all honest nodes can commit their random bits to the witness committee correctly, we also assume a 1-late adversary (Definition~\ref{def:late_adversary}), meaning it does not know honest nodes' random bits before they are committed. Under these conditions, we obtain the following theorem (a simplified version of Theorem~\ref{thm:common_coin}).

\begin{theorem}
	\label{thm:common_coin_simplified}
	Given a system of witness committees as in Definition \ref{def:witness_committees}. We can compute a common random coin in the synchronous setting with a 1-late adversary without using signatures with $\tilO(\delta)$ total communication workload per node in $\bigO(\log_\delta n)$ rounds.
\end{theorem}

Another application of our reliable aggregation protocol using a system of witness committees is a solution for the multi-value consensus problem with the desired characteristics. In the consensus problem each node proposes one of $k$ values and eventually all must decide one and the same value. If all honest nodes propose the same value, then this value must be decided upon (see Definition \ref{def:consensus}). What makes this problem hard is that each node is forced decide a value eventually. By comparison, in the reliable broadcast problem we only require a value to be delivered if the sender is honest (else nothing may ever happen).

The consensus problem is famously known to be intractable in the asynchronous, deterministic setting \cite{DBLP:journals/jacm/FischerLP85} and even has a $\widetilde{\Omega}\big(\!\sqrt{n}\big)$ round lower bound in the synchronous, randomized setting where a constant fraction of nodes is Byzantine \cite{DBLP:conf/podc/Bar-JosephB98}. However, given that we \emph{already have} a system of witness committees from a pre-computation step, we can actually bypass these lower bounds and solve the consensus problem with $<n/2$ Byzantine nodes deterministically in $\tilO(1)$ rounds with $\tilO(1)$ work per node (as many times as we desire). This means that we can shift the hard requirements of solving consensus regarding synchrony, randomization and adversary threshold onto the pre-computation of witness committees.

Roughly, we use the reliable aggregation as subroutine to count the number of different proposals in the network. Knowing that each node obtains the same count of proposals from the reliable aggregation is immensely helpful, since if we see an absolute majority we can immediately decide it since it must have been proposed by at least one honest node. Else we can be sure that not all honest nodes proposed the same value, and decide a ``default'' value. The following theorem is a simplified version of Theorem \ref{thm:consensus}.

\begin{theorem}
	\label{thm:consensus_simplified}
	Given a system of witness committees as in Definition \ref{def:witness_committees} and $t<\tfrac{n}{2}$ Byzantine nodes, we can solve the consensus problem in the synchronous setting without using signatures with $\tilO(\delta)$ total communication workload per node in $\bigO(\log_\delta n)$ rounds.
\end{theorem}

\subsection{Preliminaries}

\textbf{Synchrony assumptions.} In a \emph{synchronous network} we assume that there are known and fixed bounds on message delivery times, meaning all events occur within predictable time limits. \emph{Partial synchrony} relaxes these assumptions by acknowledging that such bounds may exist but are unknown or only hold after some unknown global stabilization time. \emph{Asynchrony}, in contrast, imposes no assumptions on timing, allowing for arbitrary and unbounded delays in message delivery, which is the most challenging setting as it is impossible to distinguish between slow and failed components \cite{DBLP:journals/jacm/FischerLP85}.

\textbf{Adversary Limitations.} We consider a Byzantine adversary that controls a fixed set $V_b \subseteq V$ of up to $t$ nodes, for some $t \leq n$. The actual number of Byzantine nodes $|V_b| \leq t$ is unknown to honest nodes. The general goal of this work is to provide resilience for a linear number of adversarial nodes, i.e., $t \in \Theta(n)$. 
In the following, we discuss different limitations of an adversary that controls Byzantine nodes.
It is important to note that we will not require all of these limitations simultaneously for any of the algorithms we consider. Further, whenever we set (reasonable) limitations on the adversary, these are typically strictly required to enforce an algorithmic idea or due to known impossibility results.

\begin{definition}[Adversary Share of Bandwidth]
	\label{def:msg_restricted}
	Assume a synchronous network and let $\sigma$ be the (minimum) number of bits that any honest node is capable to send per round. We define the honest bandwidth as $B = n \cdot \sigma$, which describes the number of bits that honest nodes $V_h$ can send each round assuming all are honest. For $b \in [0,1)$ we call an adversary $b$-bandwidth restricted if Byzantine nodes have at most a $b$-fraction of the bandwidth $B$, that is, all nodes $V_b$ send at most $b \cdot B$ bits per round \emph{in total}.
\end{definition}

\begin{definition}[Late Adversary]
	\label{def:late_adversary}
	Let $r$ be the current round. We call an adversary $\ell$-late if it knows the local states of all honest nodes (including local random coins) that were computed in round $r-\ell$ and earlier (smaller round numbers), while being oblivious to the local state of honest nodes computed in round $r-\ell+1$ and later (i.e., larger round numbers). For our purpose it suffices to assume an adversary that is 1-late. The 1-late adversary is motivated by the fact that communicating any local information from an honest node to this adversary takes at least one communication round.
\end{definition}

\textbf{Probabilistic Definitions and Notation.} 
Randomized Byzantine fault-tolerant protocols, particularly for consensus, often aim for Las-Vegas style algorithms that succeed with probability 1 as time goes to infinity. These algorithms cannot provide fixed running time bounds, as this would essentially imply that they are in fact deterministic. Typically, only an expected duration can be provided. However this does not guarantee termination with high probability within a reasonable time, in particular if distribution of the running time has a long tail. Therefore, we focus on Monte-Carlo algorithms, which succeed within a \textit{fixed} time bound with probability strictly less but close to 1.

The success probability $1 - \tfrac{1}{n^c}$ for some constant $c>0$ is known as \emph{with high probability} (w.h.p., see Definition \ref{def:whp}) and ensures high success probability for large networks. It is commonly used in the literature for its convenience in asymptotic analysis. First, it allows applying union bounds over $\poly n$ ``bad'' events, showing that none will occur w.h.p. Second, it ties performance and success probability to the single parameter $n$, whereas $c$ is typically dropped in the O-notation. However, this also obscures any dependency on $c$! Specifically, ``nested'' applications of union bounds cause algorithm performance to scale very badly in $c$ requiring $c$ to be rather small for practical use. Therefore, a minimum success probability can only be guaranteed if $n$ is relatively large, which is problematic for small networks.

To address this, some randomized consensus protocols \cite{DBLP:journals/corr/abs-1906-08936}\ps{check out if there are more, maybe algorand?} use a separate \emph{security parameter} $\lambda$, to provide a baseline success probability independently of $n$, in particular for the case that $n$ is small. The goal is to ensure that the failure probability is \textit{negligible} as function of $\lambda$ (e.g., $e^{-\lambda}$, see Definitions \ref{def:negligible_function} and \ref{def:all_but_neglible_probability}) but performance stays acceptable in terms of $\lambda$, which is only possible by making the required base-line success probability explicit via $\lambda$.
However, this concept is flawed as well. In particular, since union bounding over $\poly n$ bad events is usually \textit{necessary}, one has to implicitly assume a polynomial dependence of $\lambda$ on $n$, which obscures the impact of $n$ on the correlation between performance and success probability!

The bottom line is that the performance of Monte Carlo algorithms fundamentally depends on \emph{both} the guaranteed base line success probability provided by $\lambda$ \emph{and} the network size $n$, and we aim to study the effect of both parameters. We simply combine the two concepts by requiring our algorithms to succeed both w.h.p.\ and with all but negligible probability, by using the maximum of both probabilities. Throughout this paper we will use the term \textit{with high confidence (w.h.c.)} for this concept which we formalize as follows.

\begin{definition}[With High Confidence]
	\label{def:whc}
	We say an event $E$ occurs with high confidence (w.h.c.) with respect to security parameter $\lambda$ and network size $n$ if it occurs both\textit{ with high probability} (according to Definition \ref{def:whp}) and \textit{all but negligible probability} (Definition  \ref{def:all_but_neglible_probability}) with respect to $\lambda$. More specifically, $\pr(E) \geq 1-\min\big(\frac{1}{n^c}, \text{negl}(\lambda)\big)$ for $c > 0$ and a function $\text{negl}(\lambda)$ that is negligible in $\lambda$ (see Definition \ref{def:negligible_function}).
\end{definition}

On a mathematical level the corresponding analysis is not much more complicated besides an additional parameter $\lambda$ showing up allowing us to optimize performance for small networks as well (the simplified theorems in the contributions section assume $\lambda \in \tilO(1)$). The following lemma can be seen as an interface between the three concepts, as it gives us a threshold probability such that an event holds  \textit{w.h.p.} \textit{and} \textit{with all but negligible probability}, and therefore \textit{w.h.c.} Besides this lemma and a specialized union bound (Corollary \ref{cor:unionbound_whc}) no additional overhead is incurred from using w.h.c.\ instead of  the more traditional w.h.p.

\begin{lemma}
	\label{lem:whc}
	Let $E$ be an event that occurs with probability $\pr(E) \geq 1-e^{-\mu}$, for $\mu \geq c \ln n + \lambda$ and $c > 0$. Then $E$ occurs w.h.c.
\end{lemma}

\begin{proof}
	The probability that $E$ does not occur is \[e^{-\mu} \leq e^{-(c\ln n + \lambda)} = e^{-c\ln n} \cdot e^{-\lambda} = \tfrac{1}{n^c}\cdot e^{-\lambda} \leq \min\big(\tfrac{1}{n^c},e^{-\lambda}\big). \tag*{\qedhere}\] %
\end{proof}

Additional probabilistic concepts, explanations and bounds are provided in Appendix \ref{sec:probabilistic_concepts}.

\section{System of Witness Committees}
\label{sec:computing_committees}

Witness committees are small sets of nodes with a majority of honest members and are a fundamental building block for our solutions to distributed problems such as reliable broadcast and consensus. These committees enable the irrevocable commitment of values that can be retrieved and trusted for consistency, provided a sufficiently large subset (a quorum) of witnesses agrees on the value. Due to their small size, witness committees minimize the total message complexity almost to the theoretical limit of $\Omega(n)$, which compares much more favorably to the typical $\Theta(n^2)$ when relying on linear sized quorums (e.g., \cite{DBLP:journals/jacm/MostefaouiMR15, DBLP:journals/jacm/DolevR85}).

Prior work using small witness committees for distributed problems (e.g., \cite{DBLP:conf/sosp/GiladHMVZ17, DBLP:conf/opodis/BezerraAKS024}) typically relies on a single committee, requiring $\Omega(n)$ workload per witness to provide every node access to a quorum. In contrast, in Section \ref{sec:applications} we design algorithms with near-constant workload for any node by leveraging parallelization through a \emph{system} of $\Theta(n)$ witness committees that are evenly distributed across the network and approximately known by all nodes.

This approach shifts much of the required communication to the pre-computation phase of computing the system of witness committees. Learning $\Theta(n)$ committees that are (necessarily) determined via a random process, inherently requires $\Theta(n)$ messages per node as this corresponds to learning a random variable with Shannon entropy $\Omega(n)$. Our algorithms are tight with this theoretical lower bound (up to logarithmic factors). Further, we keep with the near-constant workload per round by spreading the message load over $\tilO(n)$ rounds. Importantly, as the resulting committees can be reused, this up-front cost can be amortized over repeated applications of reliable broadcast or consensus.\ps{this is redundant with the contributions section, potential for shortening.}

Each witness committee in the system is associated with a node $u$ (referred to as $u$'s committee). Though not every node has an associated committee, we aim for a constant fraction of nodes to have one. The system of witness committees is characterized by sets $C_{uv}$, which represent the local view of $u$'s committee as known by a node $v$ in the network. In this framework, each node $v$ has its own notion of $u$'s committee. The unifying property of these local views is that all sets $C_{uv}$ (for all $v$ and a fixed $u$) intersect in a shared \textit{common core} of honest nodes. Importantly, nodes $v$ do not need to explicitly identify the common core within their local view $C_{uv}$.\footnote{The term "common core" is used in the context of the gather protocol \cite{DBLP:conf/stoc/CanettiR93}, which addresses a weaker variant of the interactive consistency problem. In the gather protocol, each node $v$ proposes a value $x_v$, and all honest nodes must output subsets of pairs $(v, x_v)$ such that these subsets intersect in a sufficiently large common core. This intersection property is a defining characteristic of the sets $C_{uv}$. However, beyond this similarity, our algorithms, model assumptions, performance criteria, and objectives differ from those of the gather protocol.}\ps{maybe put footnote in related work}

Furthermore, each set $C_{uv}$ must contain fewer than one-third Byzantine nodes. Concretely, we require that its common core has size at least $\beta$, while the total size of $C_{uv}$ is less than \smash{$\tfrac{3\beta}{2}$}, for some small (near-constant) value $\beta$. It suffices that a constant fraction of nodes $u \in V$ has committees $C_{uv}$ meeting these properties. Otherwise, we set $C_{uv} = \emptyset$ (and all honest nodes agree on this), effectively marking $u$'s committee as invalid. Definition~\ref{def:witness_committees} provides the precise formalization.

\begin{definition}[Witness Committees]
	\label{def:witness_committees}
	We call $\big\{ C_{uv} \subseteq V \mid u \in V, v \in V_h \big\}$ a \textit{system of witness committees} with availability $0 < \alpha \leq 1$ and size $\beta \in \mathbb N$ if the following holds. %
	\begin{itemize}
		\item \textbf{Agreement:} For each $u \in V$, either $C_{uv} = \emptyset$ for all honest nodes $v \in V_h$, or there exists a common core $C_u \subseteq V_h$ of size at least $\beta$ that is contained in all $C_{uv}$ for $v \in V_h$.%
		\item \textbf{Membership:} Each honest node is member of witness committee $C_{uv}$ with $v \in V_h$ for at most $2\beta$ distinct nodes $u \in V$. Each committee $C_{uv}$ for $v \in V_h$ contains less than \smash{$\tfrac{3\beta}{2}$} nodes.
		\item \textbf{Availability:} For at least $\alpha \cdot n$ honest nodes $u \in V_h$  there exist an honest node $v \in V_h$ with $C_{uv} \neq \emptyset$.
	\end{itemize}
\end{definition}

The following theorem states the algorithmic properties of our solution to obtain witness committees as defined above. In the remaining part of this section  we present our algorithm in several phases and prove the properties of each phase, resulting in the proof of the theorem at the end of this section.

\begin{theorem}
	\label{thm:set_up_committee}
	A system of witness committees  as in Definition \ref{def:witness_committees} with size $\beta \in \bigO(\log n + \lambda)$ and availability \smash{$\alpha = \tfrac{1}{6}$} can be computed in $\bigO(n)$ rounds and $\sigma \in \bigO\big(\poly{\log n + \lambda}\big)$ bit bandwidth per node and round, w.h.c. The algorithm works in the synchronous setting and tolerates a Byzantine adversary that controls $t$ nodes and is $b$-bandwidth restricted with $t,b \leq n/24$. \ps{maybe give exact degree of the polynomial}
\end{theorem}

Our algorithm consists of three phases (A,B and C). In each phase we compute a (preliminary) system of committees that are named analogously ($A_{uv}$, $B_{uv}$ and $C_{uv}$). With each phase the respective (preliminary) committees satisfy increasingly stricter conditions culminating in the final system of witness committees $C_{uv}$ after phase C. We use various variables and parameters during each phase, a glossary is provided in Appendix \ref{sec:variable_glossary}.

For each phase we present pseudo-code that is streamlined using the assumption of a synchronous network model. Specifically, since all nodes know an upper bound on the number of rounds each section requires, nodes are assumed to start these sections simultaneously. Some sections only involve local computations with no message exchanges, which complete instantly in our model. Other sections send multiple messages and thus require several rounds.

Our notation ``\textbf{send}~\dots~\textbf{to}~\dots~\textbf{foreach}~\dots'' specifies a set of messages parameterized by content and receiver, without dictating how messages are being transmitted. Although, we do not specify the sending order, we assume honest nodes utilize their full bandwidth of $\sigma$ bits per round. Further, we may assume honest nodes randomize their transmissions so that each node receives at most $\sigma$ bits per round in expectation from honest nodes.

\subsection{Phase A: Construct Preliminary Committees}

In a nutshell, in Phase A each node $v \in V$ samples a random set $M_v$ of size $\gamma$ and announces to everybody that it is part of $u$'s committee for each $u \in M_v$. Then each node $v \in V$ locally constructs $A_{uv}$ as the set of all nodes that announced their membership in $u$'s committee within a certain time bound. If $A_{uv}$ becomes too large (bigger than \smash{$\tfrac{5\beta}{4}$}, where $\beta$ is proportional to $\gamma$), then it is set to $A_{uv} = \emptyset$ (invalid).

\begin{algorithm}[ht]
	\caption{\textbf{Phase A --- Construct Preliminary Committees $A_{uv}$}}
	\label{alg:PhaseA}
	\SetKwInOut{Input}{input}
	\SetKwInOut{Output}{output}
	\SetKw{KwSend}{send}
	\SetKw{KwTo}{to}
	\SetKw{KwForEach}{foreach}	
	\newcommand{\SendToFor}[3]{\KwSend\ #1\ \KwTo\ #2\ \KwForEach\ #3}
	\SetKwProg{ForSend}{foreach}{ send}{end}
	\DontPrintSemicolon
	
	\Input{%
		Sampling size $\gamma \in \Theta(\log n + \lambda)$, with $\beta \in \mathbb{N}$ derived from $\gamma$ via Lemma \ref{lem:size_committees}%
	}
	\Output{%
		Each node $v$ learns $A_{uv}$ for each node $u$
	}
	
	\vspace*{0.5\baselineskip}
	
	$v$ samples set $M_v \subseteq V$ with $|M_v| = \gamma$ uniformly at random 
	\label{line:sample_committee}
	\Comment{\textit{$v$ joins $\gamma$ random committees}}
	
	$\tau \gets $ current round
	
	\vspace*{0.25\baselineskip}
	
	\SendToFor{$\langle \texttt{member},u \rangle$}{$w$}{$(u,w) \in M_v \times V$}\Comment{\textit{$v$ announces to join $u$'s committee}}

	\vspace*{0.25\baselineskip}
	
	\ForEach{$u \in V$\label{line:loop_prelim_committee}}
	{
		$A_{uv} \gets \big\{w \in V \mid v \text{ received } \langle \texttt{member},u\rangle \text{ from } w \text{ until round } \tau + \lceil \gamma n/\sigma \rceil \big\}$
		\label{line:time_bound}\label{line:draft_prelim_committee}
		
		\lIf(\Comment{\textit{if too many malicious nodes join $A_{uv}$, it is invalid}}){$|A_{uv}| \geq \tfrac{5\beta}{4}$}
		{
			$A_{uv} \gets \emptyset$ \label{line:prelim_committee_empty}
		}
	}
	
\end{algorithm}

We will see that the minimum condition of the (preliminary) committees over all phases A,B and C is that for some fixed node $u$ all \emph{non-empty} committees $A_{uv}, B_{uv}, C_{uv}$ contain the same common core $C_u$ of honest nodes defined below. The common core $C_u$ is the characterizing property of $u$'s committee and whenever we speak of $u$'s committee we usually mean $C_u$.

\begin{definition}[Common Core]
	\label{def:common_core}
	For some $u \in V$, the common core of $u$'s committee is the union of all honest nodes in all $A_{uv}$, formally $C_u = \bigcup_{v \in V_h} A_{uv} \cap V_h$.
\end{definition} 

To see that the common core $C_u$ contains nodes that all $A_{uv} \neq \emptyset$ have in common, 
we show that any honest node $v \in V_h$ that wants to join $u$'s committee (i.e., $u \in M_v$ in Line \ref{line:sample_committee}), will be member of all non-empty preliminary committees $A_{uw}$ of honest nodes $w \in V_h$.

\begin{lemma}	
	\label{lem:include_honest_nodes_in_committees}
	Let $v \in V_h$ be an honest node. Then $u \in M_v$ if and only if $v \in A_{uw}$ for all $w \in V_h$ with $A_{uw} \neq \emptyset$.
\end{lemma}

\begin{proof}
	Since $v$ is honest and the time bound in Line \ref{line:time_bound} is sufficient for $v$ to send all its messages $\langle \texttt{member},u \rangle$ in the synchronous setting, any honest node $w \in V_h$ will eventually receive this message from $v$ and add $v$ to $A_{uw}$. Since the receiver of a message can identify the sender this is also the only way that $w$ will add $v$ to $A_{uw}$ (no Byzantine node can cause another node to become member of any committee).
	The only condition that removes $v$ from $A_{uw}$ occurs when $A_{uw}$ exceeds the size limitation in Line \ref{line:prelim_committee_empty}, in which case we have $A_{uw} = \emptyset$.
\end{proof}

 Lemma \ref{lem:include_honest_nodes_in_committees} implies that all non-empty preliminary committees $A_{uv}$ of honest nodes $v \in V_h$ contain exactly the same set of honest nodes. This observation implies the following Lemma.

\begin{lemma}
	\label{lem:common_core}
	We can equivalently characterize the common core $C_u$ from Definition \ref{def:common_core} as follows $C_u = \bigcap_{v \in V_h, A_{uv}\neq \emptyset} A_{uv} \cap V_h$.	
\end{lemma}

\begin{proof}
	It is clear that $\bigcap_{v \in V_h, A_{uv}\neq \emptyset} A_{uv} \cap V_h \subseteq C_u$. It remains to show $C_u \subseteq\bigcap_{v \in V_h, A_{uv}\neq \emptyset} A_{uv} \cap V_h$ to prove the equality. 	
	Let $w \in C_u := \bigcup_{v \in V_h} A_{uv} \cap V_h$.
	Lemma \ref{lem:include_honest_nodes_in_committees} implies that if there is an honest node $w \in V_h$ with $w \in A_{uv}$ then $w \in A_{ux}$ for all honest nodes $x \in V_h$ with $A_{ux} \neq \emptyset$. Thus $w \in  \bigcap_{v \in V_h, A_{uv}\neq \emptyset} A_{uv} \cap V_h$. 	
	For the border case $A_{uv} = \emptyset$ for all $v \in V_h$ we have $C_u = \emptyset$. So here $ \bigcap_{v \in V_h, A_{uv}\neq \emptyset} A_{uv} \cap V_h$ is an ``empty intersection'' which we formally consider as the empty set $\emptyset$ as well.
\end{proof}

The lemma above already provides a preliminary version of the agreement property of Theorem \ref{thm:set_up_committee} since now the nodes that have a non-empty preliminary committee $A_{uv} \neq  \emptyset$ for $u$, agree on a common core $C_u$. 

Looking ahead to later phases, we will see that if for $u \in V$ the common core $C_u$ is ``supported'' by many nodes $v \in V_h$ (by which mean $A_{uv} \neq \emptyset$) then $C_u$ can be disseminated to all honest nodes who do not support $C_u$ yet. For the rest of nodes $u \in V$ we will establish the condition that \emph{all} committees of honest nodes are made empty, i.e., no honest node will support $u$'s committee and it is considered ``invalid'' (cf.\ Definition \ref{def:witness_committees}).

The next Lemma gives numerical bounds on the number of honest nodes in preliminary committees that are non-empty. In particular, it shows that for each node $v \in V$ the preliminary committee $A_{uv} \neq \emptyset$ contains at most \smash{$\tfrac{5\beta}{4}$} nodes overall, of which at least  $\beta$ are honest but less than \smash{$\tfrac{9\beta}{8}$} honest, w.h.c.

\begin{lemma}
	\label{lem:size_committees}
	For each honest node $v \in V_h$ and for all $u \in V$ we have that \smash{$|A_{uv}| < \tfrac{5\beta}{4}$}. Further, if the preliminary committee $A_{uv} \neq \emptyset$ is non-empty then $A_{uv}$ contains at least $\beta$ but less than \smash{$\frac{9\beta}{8}$} \emph{honest} nodes, for \smash{$\beta := \lceil \frac{9\gamma(n-t)}{10n} \rceil \in \big[ \tfrac{\gamma}{2}, \gamma \big] $} and a suitably chosen value $\gamma \in \bigO(\log n + \lambda)$, w.h.c.
\end{lemma}

\begin{proof}
	The condition \smash{$|A_{uv}| < \tfrac{5\beta}{4}$} is enforced in Line \ref{line:prelim_committee_empty}.	
	Consider a committee $A_{uv} \neq \emptyset$ %
	and let $w \in V_h$ be an honest node. Let the random variable $X_{w} = 1$ if $w$ joins $A_{uv}$ else $X_{w} = 0$. Note that we did not exclude the possibility that $w \in A_{uw}$ or even $w \in A_{wv}$, i.e., $w$ can (randomly) decide to join its own committee or committees hosted by itself (see Line \ref{line:sample_committee} of Algorithm \ref{alg:PhaseA}), thus the probability that $w$ joins $A_{uv}$ is $\frac{\gamma}{n}$.
	
	Therefore, the random variables $X_{w} \sim \text{Ber}(\tfrac{\gamma}{n})$ for $w \in V_h$, are Bernoulli variables with distribution $\tfrac{\gamma}{n}$. They are independent, since different honest nodes $w_1,w_2 \in V_h$ do not influence each others random sample. Then, $X := \sum_{w \in V_h} X_{w}$ is the random number of honest nodes that join $A_{uv}$, with expected value \smash{$\ex(X) \geq \frac{\gamma (n-t)}{n}$}. This means that $\beta = \lceil \tfrac{9}{10}\cdot \ex(X) \rceil$. By the Chernoff bound given in Lemma \ref{lem:chernoffbound} and a suitable constant $d >0$ we obtain
	\begin{align*}
		& \pr\big(X \geq \tfrac{9\beta}{8} \big) = \pr\big(X \geq \tfrac{9}{8}\cdot \big\lceil \tfrac{9\ex(X)}{10}   \big\rceil\big) \leq \pr\big(X \geq (1+ \tfrac{1}{80})\ex(X)  \big) \leq e^{-d \cdot \ex(X)}, \\
		& \pr\big(X < \beta \big) = \pr\big(X < \big\lceil \tfrac{9\ex(X)}{10}   \big\rceil\big) \leq\pr\big(X \leq (1-\tfrac{1}{10}) \ex(X)\big) \leq e^{-d \cdot \ex(X)}
	\end{align*}
	By choosing \smash{$\gamma \geq  \tfrac{1}{d}(c\ln n + \lambda)\frac{n}{n-t}$} (for some $c>0$), we have that \smash{$\ex(X) \geq \tfrac{1}{d}(c\ln n + \lambda)$}. Therefore, the number of honest nodes in $A_{uv}$ is at least $\beta$ and less than \smash{$\tfrac{9\beta}{8}$} with probability at least $1-e^{\mu}$ for $\mu = c\ln n + \lambda$. By Lemma \ref{lem:whc} the claim holds w.h.c. Note that $\gamma \in \bigO(\log n + \lambda)$, since $\frac{n}{n-t} \in O(1)$ due to $t < n/3$. Finally, the claim is true for all $u \in V$ by invoking a union bound, see Corollary \ref{cor:unionbound_whc}.%
\end{proof}

This lemma has implications for the size of the common core $C_u$ that we introduced earlier, in particular, it is very likely to have size between $\beta$ and \smash{$\tfrac{9\beta}{8}$} (unless all preliminary committees are empty). 

\begin{lemma}
	\label{lem:common_core_size}
	Let $u \in V$ with common core $C_u \neq \emptyset$. Then $\beta \leq |C_u| < \tfrac{9\beta}{8}$, w.h.c.
\end{lemma}

\begin{proof}
	In this proof we assume that the claim outlined in Lemma \ref{lem:size_committees} is true, which is the case w.h.c. In particular, the Lemma is true for all nodes $v \in V$ w.h.c., by invoking a union bound, see Corollary \ref{cor:unionbound_whc}.
	Let $u \in V_h$ with $A_{uv} \neq \emptyset$, which exists due to $C_u \neq \emptyset$. By Lemma \ref{lem:size_committees}, we have \smash{$|A_{uv}\cap V_h| \geq \beta$} and \smash{$|A_{uv}|  < \tfrac{9\beta}{8}$}, w.h.c. Since $A_{uv} \cap V_h \subseteq C_u$ by Definition \ref{def:common_core} %
	we obtain $|C_u| \geq \beta$. 
	
	For a contradiction assume \smash{$|C_u| \geq \tfrac{9\beta}{8}$}. This means (a) that a single preliminary committee was too large, i.e., \smash{$|A_{uv} \cap V_h| \geq \tfrac{9\beta}{8}$ for some $v \in V_h$}. Or, (b) \smash{$|A_{uv} \cap V_h| < \tfrac{9\beta}{8}$} for all $v \in V_h$ but the union is too large \smash{$|\bigcup_{v \in V_h} A_{uv} \cap V_h| \geq  \tfrac{9\beta}{8}$}, which implies some honest cores $A_{uv} \cap V_h, v \in V_h$ must be unequal, i.e., there are three honest nodes $v_1,v_2,w \in V_h$ with $w \in A_{uv_1}$ and $w \notin A_{uv_1}$. Case (a) contradicts Lemma \ref{lem:size_committees}, while case (b) contradicts Lemma \ref{lem:include_honest_nodes_in_committees}.
\end{proof}

For the next lemma we require the following definition.

\begin{definition}[Support of $u$]
	We call the set of honest nodes $w \in V_h$ that have non-empty preliminary committees $A_{uw} \neq \emptyset$ the \emph{support} of $u$, which we denote with $S_u$.
\end{definition}

We show that a constant fraction of nodes $u \in V$ have a \emph{large support}, that is, for most honest nodes $v \in V_h$ we have $A_{uv} \neq \emptyset$. This is in preparation for the next phases where we show that such nodes with large support one can ``spread'' the knowledge of preliminary committees to all honest nodes. This fraction essentially determines the availability $\alpha$ of witness committees (cf.\ Definition \ref{def:witness_committees}).

Note that this lemma is designed for a $b$-bandwidth restricted adversary for some constant $b < 1/24$ that controls at most $n/12$ nodes.  %
We remark that the bound on $b$ and $t$ is not fully optimized and we can obtain better bounds by being more careful, which we ignore for the sake of a simpler presentation. For the same reason we assume $\sigma$ divides $n\gamma$ (the claim holds for any $\sigma \in \tilO(1)$ given $b < \tfrac{1}{24 + \varepsilon}$ with arbitrarily small $\varepsilon > 0$ for sufficiently large $n$).

\begin{lemma}
	\label{lem:min_valid_prelim_committees_a}
	Assume $t \leq n/12$ and that $\sigma$ is chosen such that it divides $n\gamma$. Consider a $b$-bandwidth restricted adversary for $b\leq 1/24$. Then for at least $\tfrac{n}{6}$ nodes $u \in V_h$ we have $|S_u| \geq \tfrac{n}{2}$, w.h.c.
\end{lemma}

\begin{proof}
	Recall that by Lemma \ref{lem:size_committees} each $A_{uv}$ contains less than $\tfrac{9\beta}{8}$ honest nodes, w.h.c. To convince some honest node $v \in V_h$ to set $A_{uv} \gets \emptyset$ (i.e., invalidate $A_{uv}$) in Line \ref{line:prelim_committee_empty}, the adversary needs to cause \smash{$|A_{uv}| \geq  \frac{5\beta}{4}$}. For this to happen, at least \smash{$\frac{\beta}{8}$} messages $\langle \texttt{member},w \rangle$ must be sent from \emph{Byzantine} nodes $w \in V_b$ to $v$. %
	We combine the following observations to bound the number of preliminary committees $A_{uv}$ the adversary can invalidate in this way. 
	
	(i): the adversary can send at most $b \cdot n \cdot \sigma$ messages per round via Byzantine nodes. (ii): there is a round limit of $\lceil \gamma  n/\sigma\rceil$ enforced in Line \ref{line:time_bound} for declaring membership via $\texttt{member}$ messages. (iii): the adversary has to send at least \smash{$\tfrac{\beta}{8}$} messages to some node $v$ to invalidate a single preliminary committee $A_{uv}$. (iv): due to $t \leq \frac{n}{12}$ we have \smash{$\beta = \lceil \frac{9\gamma(n-t)}{10n} \rceil \geq \lceil \frac{4\gamma}{5} \rceil $}. 
	
	From (i)-(iii) we obtain that the total number of preliminary committees $A_{uv}$ the adversary can invalidate is at most $\frac{8}{\beta} \cdot  b \cdot n \cdot \sigma \cdot  \big\lceil  \frac{\gamma n}{\sigma}\big\rceil$. Furthermore,
	\begin{align*}
		\tfrac{8}{\beta} \cdot  b \cdot n \cdot \sigma \cdot  \big\lceil  \tfrac{\gamma n}{\sigma}\big\rceil 
		& \hspace*{0.5mm} \leq 	 \tfrac{8}{\beta} \cdot  b \cdot \gamma \cdot n^2  \tag*{\small since $\sigma$  divides  $n\gamma$}	\\
		& \stackrel{\text{(iv)}}{\leq} 10 \cdot  b \cdot \gamma \cdot n^2  
		 = 10 \cdot  b \cdot n^2.
	\end{align*}
	Assume the claim of the lemma is false, i.e., that for more than $\tfrac{5n}{6}$ nodes $u \in V$ we have $A_{uv} \neq \emptyset$ for less than $\tfrac{n}{2}$ honest nodes $v \in V_h$. This means that the adversary has invalidated at least \smash{$\tfrac{n}{2} \cdot \tfrac{5n}{6}  = \tfrac{5n^2}{12}$} preliminary committees $A_{uv}$. This implies the following contradiction
	\begin{align*}
		10 \cdot  b \cdot n^2 & > \tfrac{5n^2}{12}\\
		\Longleftrightarrow \quad b & > \tfrac{1}{24}. \tag*{\qedhere}
	\end{align*}	
\end{proof}

\subsection{Phase B: Refine Witness Committees} 

To motivate Phase B, we compare the conditions left by Phase A with the requirements of Phase C. From Phase A and Lemma~\ref{lem:min_valid_prelim_committees_a}, a bandwidth-restricted adversary ensures there is a fraction of nodes $u \in V$ whose support $S_u$ has size at least $\tfrac{n}{2}$. In Phase~C, we plan to exploit these large supports: if $S_u$ exceeds the number of Byzantine nodes, any honest node that has not yet learned the common core $C_u$ can reconstruct it by sampling the preliminary committees of a small number of other nodes.

\noindent
However, if $S_u$ is too small, the adversary and randomness can unduly influence a node's view of the common core. Worse, if $|S_u|$ lies in a critical range, the adversary might determine the composition of committees sampled in Phase~C, undermining the desired properties. Consequently, our goal for Phase~B is decide in advance whether $S_u$ is sufficiently large, in which case we mark $u$'s committee valid, else invalid. Our main technique for this is that the core $C_u$ determines their support $S_u$ and then forms a consensus on $u$'s validity.

\begin{algorithm}[h]
	\caption{\textbf{Phase B --- Refine $A_{uv}$ into $B_{uv}$}}
	\label{alg:PhaseB}
	
	\SetKwInOut{Input}{input}
	\SetKwInOut{Output}{output}
	
	\SetKw{KwSend}{send}
	\SetKw{KwTo}{to}
	\SetKw{KwForEach}{foreach}	
	\newcommand{\SendToFor}[3]{\KwSend\ #1\ \KwTo\ #2\ \KwForEach\ #3}
	\SetKwProg{ForSend}{foreach}{ send}{end}
	
	\DontPrintSemicolon
	
	\Input{%
		Preliminary committees $A_{uv}$ from Phase A
	}
	\Output{%
		Each node $v$ learns $B_{uv}$ for each node $u$
	}
	
	\vspace*{0.25\baselineskip}
	
	\SendToFor{$\langle \texttt{support},u, A_{uv} \rangle$}{$w$}{$u \in V$, $w \in A_{uv}$}\label{line:phase_B_start}\Comment{\textit{$v$ supports committee of $u$}}

	\vspace*{0.25\baselineskip}
	
	\ForEach(\Comment{\textit{$M_v:$ set of nodes $u$ with $v \in C_u$ (cf. Phase A)}}){$u \in M_v$}
	{
		$\texttt{supported}[u] \gets \big\{w \in V \mid v \text{ received } \langle \texttt{support},u, A_{uw}\rangle \text{ from } w \big\}$
		
		$s[u] \gets |\texttt{supported}[u]|$
		\Comment{\textit{number of nodes expressing support for $u$}}
		
		\lIfElse{$s[u] \geq \tfrac{n}{3} + t$}
		{
			$b_{uv} \gets 1$
		}
		{
			$b_{uv} \gets 0$
		}
		\label{line:prelim_committee_empty_b}
		
		$B'_{uv} \gets \big\{ x \in V \,\big\vert\, |\{ w \in \texttt{supported}[u] \mid x \in A_{uw}\}|  \geq s[u] - t \big\}$
		\label{line:set_B_start}
		
		\lIf{$s[u] \leq 4t$}
		{
			$B'_{uv} \gets \emptyset$ 
			\label{line:set_B_end}
		}
	}
	
	\vspace*{0.25\baselineskip}
	
	\ForEach{$u \in M_v$ with $\tfrac{n}{3} \leq s[u] < \tfrac{n}{3} + 2t$ in parallel \label{line:validity_consensus_condition}}
	{        
		Participate in consensus protocol among nodes in $B'_{uv}$ on the bits $b_{uv}$
		\label{line:validity_consensus}
		\Comment{\textit{e.g., \cite{DBLP:journals/siamcomp/GarayM98}}}
	}
	
	\vspace*{0.25\baselineskip}
	
	\SendToFor{$\langle \texttt{valid},u,b_{uv} \rangle$}{$w$}{$(u,w) \in M_v \times V$}\label{line:send_validity_bit}	\Comment{\textit{broadcast validity of $u$'s committee}}
	
	\vspace*{0.25\baselineskip}

	\ForEach{$u \in V$\label{line:send_validity_bit_start}}
	{
		$\texttt{valid}[u] \gets \big|\big\{x \in A_{uv} \mid v \text{ received } \langle \texttt{valid},u,1\rangle \text{ from } x \big\}\big|$
		\Comment{\textit{votes that $A_{uv}$ is valid}}
		\label{line:validity_count}
		
		\lIfElse{$\texttt{valid}[u] \geq \beta$}
		{
			$B_{uv} \gets A_{uv}$
		}
		{
			$B_{uv} \gets \emptyset$
		}
		\label{line:validity_decision}
	}
	
\end{algorithm}

Let us look at the algorithm in detail. First, each node $v \in S_u$ signals its support of $u$'s committee to each node in $A_{uv}$ (which contains $C_u$).
If the recorded support is sufficiently large ($\geq \tfrac{n}{3}+t$), then there will be agreement among $C_u$ that $u's$ committee is valid and if it is too small ($< \tfrac{n}{3}$) there will be agreement that it is invalid. In between we do not care as long as there is consensus, which we enforce by running classical synchronous consensus protocols on the small cores $C_u$ (in parallel for each $u \in V$).

The core $C_u$ then reports back the consensus on $u$'s validity to all nodes $v \in V$. Based on this, the honest nodes $v \in V_h$ can then compute the next preliminary committees $B_{uv}$ such that the following holds. If $S_u$ is smaller than $\tfrac{n}{3}$ we have $B_{uw} = \emptyset$ for all honest nodes $w \in V_h$. Else, if $S_u$ is larger than $\tfrac{n}{3}+t$ we have $B_{uw} = A_{uw}$ for all honest nodes. In between it can be either $B_{uw} = A_{uw}$ or $B_{uw} = \emptyset$, but the result is consistent among all $B_{uw}$.

We start with the first proof that refers to Lines \ref{line:phase_B_start}-\ref{line:prelim_committee_empty_b}. First, all nodes $v \in V_h$ send their message of support for $u$'s committee to their respective preliminary committees $A_{uv}$ (with $C_u \subseteq A_{uv}$ being the intended recipient). The nodes in $C_u$ count the support messages in a variable $s[u]$, which is at least as large as the actual size of $u$'s support $|S_u|$.  We then initialize a bit-flag $b_{uv}$, with $b_{uv} = 1$ if $s[u]$ is bigger than $\tfrac{n}{3} +t$. 

Note that $s[u]$ is only an upper bound on $|S_u|$ and subject to adversary influence, so the $b_{uv}$ might be unequal among the $v \in C_u$! However, for the case where the actual support $S_u$ is relatively large ($\geq \tfrac{n}{3} +t$) or small ($<\tfrac{n}{3}$), the initialization of $b_{uv}$ is enough to obtain agreement on $b_{uw} = 1$ or $b_{uw} = 0$, among the nodes $w \in C_u$. We show this in the following lemma.

\begin{lemma}
	\label{lem:flag_b}
	The following holds w.h.c.:
\begin{itemize}		
	\item If $|S_u| \geq \tfrac{n}{3} +t $, then $b_{uw} = 1$ for all $w \in C_u$.
	\item If $|S_u| < \tfrac{n}{3}$, then $b_{uw} = 0$ for all $w \in C_u$.
\end{itemize}	
\end{lemma}

\begin{proof}
	Let $v \in C_u$. By definition any node $w \in S_u$ is honest and has $A_{uw} \neq \emptyset$. By Lemma \ref{lem:common_core} we have $C_u \subseteq A_{uw}$. Consequently every $w \in S_u$ sends a message $\langle \texttt{support}, u, A_{uw} \rangle$ to $v$. The nodes in $S_u$ are also the only \emph{honest} nodes that send such a message to $v$. 
	
	Node $v$ collects all nodes $w \in V$ that sent a message $\langle \texttt{support}, u, A_{uw} \rangle$ to $u$ in $\texttt{supported}[u]$, thus $s[u] := |\texttt{supported}[u]| \geq |S_u|$. Since there are at most $t$ malicious nodes we can deduce $s[u] \leq |S_u| + t$. Node $v$ sets $b_{uv}$ to 1 if $s[u] \geq \tfrac{n}{3} +t$, else to 0. If $|S_u| \geq \tfrac{n}{3} +t$, then $s[u] \geq |S_u| \geq \tfrac{n}{3} +t$ thus $v$ sets $b_{uv}$ to 1. If $|S_u| < \tfrac{n}{3}$, then $s[u] \leq |S_u| +t < \tfrac{n}{3} +t$, thus $b_{uv}$ is set to 0.
	
	The consensus step in Line \ref{line:validity_consensus} does not affect the values $b_{uw}$ for $S_u$ in the given range as an already existing consensus among all honest nodes $C_u$ in the committee will be preserved.	
\end{proof}

The lemma above leaves a gap where the bits $b_{uv}$ are not (yet) in agreement among the $v \in C_u$. Specifically, if the support is in the range $\tfrac{n}{3} \leq |S_u| < \tfrac{n}{3} +t$, the adversary can decide the bit $b_{uv}$ for each node $v \in C_u$, and we need additional coordination among nodes in $C_u$ to form consensus on the values $b_{uv}$. The idea is to run a standard consensus protocol among the nodes in $C_u$.

But before we can actually do this we need to resolve the problem that the nodes in $C_u$ do not necessarily know of each other (as the preliminary sets could be $A_{uv} = \emptyset$ for any $v \in C_u$). However, since the support is large, we can compute a set $B'_{uv}$ for each $v \in C_u$ that contains $C_u$ (Lines \ref{line:set_B_start}-\ref{line:set_B_end}). %
The sets $B'_{uv}$ for $v \in C_u$ have the following properties.\ps{explain this part a bit more and emphasize as technical contribution.}

\begin{lemma}
	\label{lem:set_B}
	Let $u \in V$ and $v \in C_u$. The set $B'_{uv}$ computed in Lines \ref{line:set_B_start}-\ref{line:set_B_end} has the following properties w.h.c.:
	\begin{itemize}				 		
		\item If $B'_{uv} \neq \emptyset$, then $C_u \subseteq B'_{uv}$.
		\item If $|S_u| > 4t$, then $C_u \subseteq B'_{uv}$.%
		\item $|B'_{uv}| < \tfrac{3\beta}{2}$.
	\end{itemize}
\end{lemma}

\begin{proof}
	Since $v \in C_u$, the node $v$ is honest by definition of $C_u$, thus $u \in M_v$. Therefore $v$ executes Line \ref{line:set_B_start} for node $u$, i.e., it computes the set $B'_{uv}$ of nodes $x \in V$ that appear in at least $s[u]-t$ of the sets $A_{uw}$ it obtained via messages $\langle \texttt{support}, u, A_{uw} \rangle$, where $s[u]$ is the number from which such a support message was received.
	
	We have (1) $s[u] \geq |S_u| \geq  s[u] - t$ since all nodes in $w \in S_u$ are honest and thus send their support message $\langle \texttt{support}, u, A_{uw} \rangle$ to $v$ and at most $t$ additional Byzantine nodes $w \notin S_u$ can sent such a support message. 
	Moreover, it is (2) $C_u \subseteq A_{uw}$ for all $w \in S_u$ by Lemma \ref{lem:common_core}.	
	
	First, let us assume that $S_u \neq \emptyset$. 
	By combining $S_u \neq \emptyset$ with (1) and (2), we obtain that all nodes $x \in C_u$ satisfy the condition to be member of $B'_{uv}$ in Line \ref{line:set_B_start}. That implies $C_u \subseteq B'_{uv}$, after Line \ref{line:set_B_start} (before $B'_{uv}$ is potentially deleted in Line \ref{line:set_B_end}).
		
	Now assume that $B'_{uv} \neq \emptyset$. This implies that $B'_{uv}$ was not deleted in Line \ref{line:set_B_end}, therefore we have $s[u] > 4t$, due to Line \ref{line:set_B_end}. Since $|S_u| \geq s[u] -t > 3t$ by Equation (1), we have $S_u \neq \emptyset$, thus  $C_u \in B'_{uv}$ from the deduction above. This shows the first property.
	
	Assume that $|S_u| > 4t$. Then $S_u \neq \emptyset$ thus $C_u \subseteq B'_{uv} \neq \emptyset$ in Line  \ref{line:set_B_start}  as shown before.  Further,  $s[u] \geq |S_u| > 4t$ by Equation (1), thus the condition in Line \ref{line:set_B_end} does not trigger. This shows the second property.
	
	The third property is clear in case $s[u] \leq 4t$, since then $B'_{uv} = \emptyset$ after Line \ref{line:set_B_end}. Let us look at the case $s[u] > 4t$. Let $x \in B'_{uv}$, then we have $x \in A_{uw}$ for at least $s[u]-t$ nodes $w \in \texttt{supported[u]} \subseteq V$ by Line \ref{line:set_B_start}. Because at most $t$ nodes are Byzantine, it holds that $x \in A_{uw}$ for at least $s[u]-2t$ nodes $w \in S_u$.	
	
	We now need two additional properties: (3)  \smash{$A_{uw} < \frac{5\beta}{4}$} due to Line \ref{line:prelim_committee_empty} and 
	(4) $|C_u| \geq \beta$ w.h.c., by Lemma \ref{lem:size_committees}. Due to (2) and (3), for any $w \in S_u$ have less than \smash{$\tfrac{5\beta}{4} \m |C_u|$} nodes $x \in A_{uw}$ with $x \notin C_u$. 	
	Assume $x \in B'_{uv}$ with $x \notin C_u$. Since $x$ must occur in at least $s[u]-2t > t$ preliminary committees $A_{uw}$ for $w \in S_u$, the number of $x \in B'_{uv}$ with $x \notin C_u$ is (w.h.c.) less than
	\[
		\big(\tfrac{5\beta}{4} \m |C_u|\big) \cdot \tfrac{|S_u|}{s[u]-2t} \stackrel{(1)}{\leq} \big(\tfrac{5\beta}{4} \m |C_u|\big) \cdot \tfrac{s[u]}{s[u]-2t} \stackrel{s[u] > 4t}{\leq} \big(\tfrac{5\beta}{4} \m |C_u|\big)  \cdot \tfrac{4t}{4t-2t} = \tfrac{5\beta}{2}-2|C_u| \stackrel{(4)}{\leq} \tfrac{3\beta}{2} \m |C_u|
	\]
	Each honest node in $B'_{uv}$, i.e., $x \in B'_{uv} \cap V_h$ must be in the common core $C_u$ (by Lemma  \ref{lem:include_honest_nodes_in_committees}). Therefore the overall number of nodes in $B'_{uv}$ can be bounded as
	\[
		|B'_{uv}|  = |B'_{uv} \cap V_h| + |B'_{uv} \cap V_b| < |C_u| +  \tfrac{3\beta}{2}-|C_u| = \tfrac{3\beta}{2},
	\]	
	 which shows the third property.	
\end{proof}

The lemma above implies that if the support is sufficiently large ($|S_u| > 4t$) then all nodes $w \in C_u$ in the common core have a set $B'_{uw}$ that contains $C_u$. Further, since $|C_u|\geq\beta$ and \smash{$|B'_{uw}|< \tfrac{3\beta}{2}$}, $B'_{uw}$ contains less than a third Byzantine nodes. This is sufficient to obtain agreement on the values $b_{uw}$ in the case $\tfrac{n}{3} \leq |S_u| < \tfrac{n}{3} +t$, where Lemma \ref{lem:flag_b} fails. 

For obtaining this agreement we can employ the standard deterministic consensus algorithm by Garay and Moses \cite{DBLP:journals/siamcomp/GarayM98} on each $C_u$ which will be fast since each $C_u$ is very small.

\begin{lemma}[from \cite{DBLP:journals/siamcomp/GarayM98}]
	\label{lem:garay}
	There is a deterministic algorithm that solves the consensus problem in the synchronous model for at most $t < n/3$ Byzantine parties in at most $t+1$ rounds with $\bigO(n^2t^2)$ bits of communication.
\end{lemma}

We show more formally that this algorithm can be simulated in parallel on each $C_u$ with sufficiently large support. This holds even if the nodes in $v \in C_u$ do not know $C_u$ exactly, but know only a slightly larger set $B'_{uv}$ that contains $C_u$. The simulation is done in parallel on each such $C_u$, which incurs a slowdown that corresponds to the maximum number $\gamma$ of cores $C_u$ that any given honest node $v \in V_h$ is member of. Note that the preconditions for the sets $B'_{uv}$ that we presume in this lemma were already established in Lemma \ref{lem:set_B} for those nodes $u$ that have sufficiently large support.

\begin{lemma}
	\label{lem:simulation_on_core}
	Let $V' \subseteq V$ and assume that for all $u \in V'$ we have a common core $C_u \subseteq V_h$ with the following properties. 
	
	\begin{itemize}
		\item $|C_u| \geq \beta$.
		\item Each honest node $v \in V_h$ is member of at most $\gamma$ cores $C_u$ for $u \in V'$.
		\item Each $v \in C_u$ knows a set  $B'_{uv} \subseteq V$ for each $u \in V'$ with $B'_{uv} \cap V_h = C_u$ and \smash{$|B'_{uv}| < \tfrac{3\beta}{2}$}.
	\end{itemize}
	
	Then we can simulate any algorithm $\mathcal A$ that is Byzantine fault tolerant for up to $t<n/3$ Byzantine parties on each of the $C_u$ in parallel with either a slowdown of a factor $\gamma$ or an communication overhead by a factor $\gamma$.
\end{lemma}

\begin{proof}

To get some intuition, consider running $\mathcal{A}$ separately within each core $C_u$. Since every honest node $v \in C_u$ knows the set $B'_{uv}$ containing $C_u$ with fewer than \smash{$\tfrac{|C_u|}{2}$} additional Byzantine nodes, each node effectively faces fewer than one-third Byzantine nodes in its local view. Thus, $\mathcal{A}$ remains correct. 

Two subtleties require attention. First, $\mathcal{A}$ typically assumes each node knows the entire set of participants, whereas here each node $v$ has a potentially different set $B'_{uv}$. However, because $\mathcal{A}$ tolerates up to one-third Byzantine nodes, any behavior by those Byzantine nodes, across different $B'_{uv}$ sets, can be seen as valid Byzantine actions within the original setting. 

Second, we must account for overlaps when running $\mathcal{A}$ on multiple committees in parallel. Since each honest node $v$ belongs to at most $\gamma$ cores, we can either (1) allow each node to handle $\gamma$-times the communication in a single round, or (2) define a ``super-round'' of $\gamma$ consecutive rounds, each dedicated to a single instance of $\mathcal{A}$. In both cases, the overhead in time or communication increases by at most a factor of $\gamma$.
\end{proof}

We can now use the  algorithm by \cite{DBLP:journals/siamcomp/GarayM98} to obtain consensus on the bits $b_{uv}$ among the nodes $v \in C_u$, where $u$'s support is in the range $\tfrac{n}{3} \leq |S_u| < \tfrac{n}{3} +t$ for which we do not have agreement yet (Lines \ref{line:validity_consensus_condition}-\ref{line:validity_consensus} in Algorithm \ref{alg:PhaseB}).

\begin{lemma}
	\label{lem:flag_b_2}
	Assume $4t < \tfrac{n}{3} \leq |S_u| < \tfrac{n}{3} +t$, then $b_{uv} = b_{uw}$ for all $v,w \in C_u$, w.h.c., after the consensus step in Line \ref{line:validity_consensus}. This step takes $\bigO(\beta)$ rounds and $\bigO(\beta^5)$ bits communication.
\end{lemma}

\begin{proof}
	For each node $w \in C_u$ it holds that  $|S_u| \leq {s}[u] \leq |S_u|  + t$ for the local variable ${s}[u]$ of $w$. Combining this with the presumption $\tfrac{n}{3} \leq |S_u| < \tfrac{n}{3} +t$ of the lemma,  implies that the condition $\tfrac{n}{3} \leq s[u] < \tfrac{n}{3} + 2t$ in Line \ref{line:validity_consensus_condition} is satisfied for all $w \in C_u$.
	Therefore, the consensus step in Line \ref{line:validity_consensus} will be executed for all $w \in C_u$.
	
	Since $|S_u| > 4t$ we can use the properties for $B'_{uv}$  established in Lemma \ref{lem:set_B} and the simulation argument in Lemma \ref{lem:simulation_on_core}, and run the consensus algorithm \cite{DBLP:journals/siamcomp/GarayM98} on each $C_u$ in parallel. In the simulation each node $v \in C_u$ uses $B'_{uv}$ as the set of participants in the algorithm. This establishes consensus $b_{uw} = b_{ux}$ for all $w,x \in C_u$. 
	
	We have $|B'_{uv}| \in \bigO(\beta)$ due to Lemma \ref{lem:set_B} and each node $v \in V_h$ is member of at most $\gamma$ cores $C_u$ (see Line \ref{line:draft_prelim_committee}) and $\gamma \in \bigO(\beta)$ (see Lemma \ref{lem:size_committees}). Therefore, Lemmas \ref{lem:garay} and \ref{lem:simulation_on_core} imply that the consensus step is completed after $\bigO(\beta)$ rounds requiring $\bigO(\beta^5)$ bits of communication per round.	
\end{proof}

We are now in a position where we have agreement on the validity of $u$'s committee among nodes in the common core $C_u$ via the bits $b_{uw}$ for $w \in C_u$, which we established in Lemmas \ref{lem:flag_b} and \ref{lem:flag_b_2}.
As last steps in Phase B (Lines \ref{line:send_validity_bit_start}-\ref{line:validity_decision}) the agreement on the validity of $u$'s committee is communicated from the core $C_u$ to all nodes in the support $S_u$ of $u$. %

Specifically, each node  $w \in C_u$ in the common core broadcasts the bit $b_{uw}$. 
Since the nodes in the support $w \in S_u$ have knowledge of preliminary committees $A_{uw}$ in which the core $C_u$ forms a majority, they can conclude the validity of $u$'s committee from the majority of the validity bits obtained.

The nodes $w \in S_u$ then construct new sets $B_{uv}$, with $B_{uv} = A_{uv}$ if $u$'s committee is valid, else $B_{uv} = \emptyset$. Consequentially, the sets $B_{uv} \neq \emptyset$ inherit all properties that we showed for the sets $A_{uv}$, with the difference that the support is now either empty or relatively large (which we require as precondition for Phase C, where we spread the $B_{uv} \neq \emptyset$ via sampling).

\begin{lemma}
	\label{lem:prelim_committees_b}
	Let $t < \tfrac{n}{12}$. Then the preliminary committees $B_{uw}$ for $u \in V$ and $w \in V_h$ computed in Phase B have the following properties w.h.c.
	\begin{itemize}		
		\item If $|S_u| \geq \tfrac{n}{3} + t$, then $B_{uw} = A_{uw}$ for all $w \in V_h$,
		\item If $|S_u| < \tfrac{n}{3}$ it is $B_{uw} = \emptyset$ for all $w \in V_h$,
		\item If $\tfrac{n}{3} \leq |S_u| < \tfrac{n}{3} + t$ it is either $B_{uw} = \emptyset$ for all $w \in V_h$ or $B_{uw} = A_{uw}$ for all $w \in V_h$.
	\end{itemize}
\end{lemma}

\begin{proof}
	
	In Line \ref{line:send_validity_bit} each node in the common core $v \in C_u$ sends a message $\langle\texttt{valid},u,b_{uv} \rangle$ to each node in $V$ with $b_{uv}$ indicating the validity of $u$'s committee.	
	Consequently, every node $w \in S_u$ obtains $\langle\texttt{valid},u,b_{uv} \rangle$ from all $v \in C_u$. In Line \ref{line:validity_count}, a given node $w \in C_u$ counts the number of nodes $x \in A_{uw}$ that regard $u$'s committee as valid, i.e., $b_{ux} = 1$ by accumulating the bits in the local variable $\texttt{valid}[u]$. Since $A_{uw} \neq \emptyset$ we know that $C_u \subseteq A_{uw}$ due to Lemma \ref{lem:common_core}. Further, $|C_u| \geq \beta$ and \smash{$|A_{uw}| < \tfrac{5\beta}{4}$}, w.h.c., due to Lemma \ref{lem:size_committees}. 
	
	If $|S_u| \geq \tfrac{n}{3} + t$, then $w \in S_u$ obtains $\langle\texttt{valid},u,1 \rangle$ from each node $C_u$ by Lemma \ref{lem:flag_b}, thus $\texttt{valid}[u] \geq \beta$. Consequently we set $B_{uw} := A_{uw}$ in Line \ref{line:validity_decision}. If $|S_u| < \tfrac{n}{3}$, then $w \in S_u$ obtains $\langle\texttt{valid},u,0 \rangle$ from each node $C_u$ (also due to Lemma \ref{lem:flag_b}), therefore $\texttt{valid}[u] < \beta$ and we set $B_{uw} := \emptyset$ in Line \ref{line:validity_decision}. Note that for any node $w \in V_h \setminus S_u$, we automatically have $B_{uw} = \emptyset$ due to $A_{uw} = \emptyset$ (see Line \ref{line:validity_decision}). 
	
	Now, if the support is in the critical range $\tfrac{n}{3} \leq |S_u| < \tfrac{n}{3} + t$, the initialization of bits $b_{uw}$ among $w \in C_u$ may depend on the adversary. The subsequent consensus step in Line \ref{line:validity_consensus} guarantees that the bits $b_{uw}$ are equal for all pairs $w \in C_u$ by Lemma \ref{lem:flag_b_2} (here we require $t < \tfrac{n}{12}$). Thus the bits  $b_{uv}$ that are sent by the $v \in C_u$ are consistent and we fall back to one of the two previously described cases. Consequently, we have either $B_{uw} = \emptyset$ for all $w \in V_h$ or $B_{uw} = A_{uw}$ for all $w \in S_u$. %
\end{proof}

Note that the lemma above implies that for any $u \in V$ the preliminary committees $A_{uv}$ and $B_{uv}$ are the same unless we set $B_{uv} = \emptyset$ for all $v \in V_h$ (i.e., we invalidated $u$'s committee in Phase B).
This implies the same for the support with respect to the $B_{uv}$ which equals $S_u$ (which we defined for the $A_{uv}$) unless $u$'s committee was invalidated in Phase B.

We slightly simplify the lemma above and combine it with Lemma \ref{lem:min_valid_prelim_committees_a} to obtain the following lemma which servers as interface  for Phase C.

\begin{lemma}
	\label{lem:min_valid_prelim_committees_b}
	Let $t < \tfrac{n}{12}$. The preliminary committees $B_{uw}$ have the following properties w.h.c.
	\begin{itemize}		
		\item For each $u \in V$ either $B_{uw} = \emptyset$ for all $w \in V_h$ or $B_{uw} = A_{uw}$ for all $w \in V_h$.
		\item For each $u \in V$ if $B_{uw} \neq \emptyset$ for some $w \in V_h$ then $|S_u| \geq \tfrac{n}{3}$.
		\item For at least $\tfrac{n}{6}$ nodes $u \in V$ there exists $w \in V_h$ such that $B_{uw} \neq \emptyset$.
	\end{itemize}
\end{lemma}

\begin{proof}
	The first property follows directly from Lemma \ref{lem:prelim_committees_b}.  The second property given in Lemma \ref{lem:prelim_committees_b} says that $|S_u| < \tfrac{n}{3}$ implies $B_{uw} = \emptyset$ for all $w \in V_h$. The contra-positive is that the existence of a $w \in V_h$ with $B_{uw} = \emptyset$ implies $|S_u| \geq \tfrac{n}{3}$, which shows the second property of this Lemma.	
	Furthermore, we have at least $\tfrac{n}{6}$ nodes $u \in V$ with $|S_u| \geq \tfrac{n}{2}$, by Lemma \ref{lem:min_valid_prelim_committees_a}. For those nodes we have $|S_u| = \tfrac{n}{2} \geq \tfrac{n}{3} + t$ and thus $B_{uv} \neq \emptyset$, which shows the third property.
\end{proof}

\subsection{Phase C: Assemble Final Committees}

\noindent
In the final Phase~C, we build on the properties of the preliminary committees $B_{uv}$ from Phase~B (summarized in Lemma~\ref{lem:min_valid_prelim_committees_b}). One key outcome of Phase~B is that, if $u$'s committee is deemed invalid, then $B_{uv} = \emptyset$ for all honest nodes $v$, ensuring agreement on the invalidity of $u$. Conversely, if $u$ is not invalid, at least $n/3$ honest nodes have $B_{uv} \neq \emptyset$. %

Leveraging these guarantees, we now construct the final witness committees $C_{uv}$. Each node $v$ samples a small random set $N_{uv} \subseteq V$ for every $u \in V$ and requests $B_{uw}$ from all $w \in N_{uv}$ (Lines~\ref{line:sample_core}, \ref{line:request_core}). In expectation, most responses come from honest nodes, while only a small fraction are Byzantine.

If $u$'s committee is ``valid'' (case~2 of Lemma~\ref{lem:min_valid_prelim_committees_b}), enough honest nodes $w \in V_h$ provide $B_{uw} \neq \emptyset$, enabling $v$ (with high confidence) to assemble a committee $C_{uv}$ that contains the core $C_u$, by collecting nodes that appear in sufficiently many of the non-empty $B_{uw}$.
By contrast, if $u$'s committee is marked ``invalid'' in Phase~B, all honest $w \in V_h$ respond with $B_{uw} = \emptyset$. Any committees returned by Byzantine nodes alone cannot cause $v$ to add any node to its final committee. Consequently, for every $u \in V$, either $C_{uw} = \emptyset$ for all honest $w \in V_h$, or $C_u \subseteq C_{uw}$.%

\begin{algorithm}[ht]
	\caption{\textbf{Phase C --- Assemble Final Committees $C_{uv}$}}
	\label{alg:PhaseC}
	\SetKwInOut{Input}{input}
	\SetKwInOut{Output}{output}
	
	\SetKw{KwSend}{send}
	\SetKw{KwTo}{to}
	\SetKw{KwForEach}{foreach}	
	\newcommand{\SendToFor}[3]{\KwSend\ #1\ \KwTo\ #2\ \KwForEach\ #3}
	\DontPrintSemicolon
	
	\Input{%
		Phase B committees $B_{uv}$, appropriate $\zeta \in \Theta(\log n + \lambda)$ (cf.\ Lemmas \ref{lem:prelim_committees_c}, \ref{lem:respond})
	}
	\Output{%
		Each node $v$ learns $C_{uv}$ for each node $u$
	}
	
	\vspace*{0.35\baselineskip}
	
	\ForEach(\Comment{\textit{sample nodes to sets $N_{uv}$ for $u \in V$ (initally empty)}}){$(w,u) \in V \times V$}
	{
		$v$ adds node $w \in V$ to set $N_{uv}$ with uniform probability \smash{$\tfrac{\zeta}{n}$}
		\label{line:sample_core}
	}
		
	\vspace*{0.25\baselineskip}
	
	\SendToFor{$\langle \texttt{request}, u\rangle$}{$w$}{$u \in V, w \in N_{uv}$}\label{line:request_core}
	
	\vspace*{0.25\baselineskip}
	
	\ForEach(\Comment{\textit{filter out nodes that sent too many requests ($\geq 2 \zeta$)}}){$w \in V$\label{line:filter_out_start}}
	{
		$\texttt{requested}[w] \gets \big\{u \in V \mid v \text{ received } \langle \texttt{request},u\rangle \text{ from } w \big\}$
		
		\lIf{$\big|\texttt{requested}[w]\big| \geq 2 \zeta$}{$\texttt{requested}[w] \gets \emptyset$}	\label{line:filter_out_end}
	}
	
	\vspace*{0.25\baselineskip}	
	
	\SendToFor{$\langle \texttt{response}, u,  B_{uv}\rangle$}{$w$}{$w \in V$ with $u \in \texttt{requested}[w]$}\label{line:respond}
	
	\vspace*{0.25\baselineskip}		
	
	\ForEach(\Comment{\textit{locally store committees of $u$ from each $w \in N_{uv}$}}){$u \in V, w \in N_{uv}$}
	{
		\lIfElse{$\langle \texttt{response}, u,  B\rangle$ received from $w$ \textbf{and} \smash{$\beta \leq |B| < \tfrac{5\beta}{4}$}}{$B_{uw} \gets B$}{$B_{uw} \gets \emptyset$}
	}
	
	\vspace*{0.25\baselineskip}

	\ForEach{$u \in V$ \label{line:assemble}}
	{

		$r[u] \gets \big|\{ w \in N_{uv} \mid B_{uw} \neq \emptyset \}\big|$
		\Comment{\textit{number of $B_{uw} \neq \emptyset$ received from nodes $w \in N_{uv}$}}
		
		$C_{uv} \gets \big\{ x \in V \,\big\vert\, |\{ w \in N_{uv} \mid x \in B_{uw}\}| 
		\geq r[u] - \tfrac{\zeta}{16} \big\}$
		\label{line:recover_core}
		
		\lIf{$r[u] \leq \tfrac{\zeta}{4}$}
		{
			$C_{uv} \gets \emptyset$
			\label{line:invalid_c}
		}
	}
	
\end{algorithm}

Before assembling the sets $C_{uv}$, we must ensure that requesting the committees $B_{uv}$ is resilient to Byzantine flooding. In particular, we must prevent adversarial nodes from overloading honest nodes with excessive requests. Since node $v$ expects at most $\zeta$ legitimate requests for committees of different $u \in V$ from any honest node $w \in V_h$, it can safely ignore nodes sending $2\zeta$ or more requests (see Lines~\ref{line:filter_out_start}--\ref{line:filter_out_end}), and still legitimate requests exchanged among honest nodes obtain a response.

\begin{lemma}
	\label{lem:respond}
	We can choose $\zeta \in \bigO(\log n + \lambda)$ such that for each $u \in V$, every $v \in V_h$ will obtain a response in the from of the committee $B_{uw}$ from each $w \in N_{uv} \cap V_h$, w.h.c.
\end{lemma}

\begin{proof}
	Fix two nodes $v,w \in V_h$. For $u \in V$ let $X_{uw}$ be indicator variable which is 1 if $v$ samples $w$ to $N_{uv}$ in Line \ref{line:sample_core}, else 0. Then \smash{$X_{uw} \sim \text{Ber}(\tfrac{\zeta}{n})$}. Let $X_w := \sum_{u \in V} X_{uw}$ be the number of $u \in V$ with $w \in N_{uv}$. As the sampling process in Line \ref{line:sample_core} is conducted independently for each pair $u,w$, this is a sum of independent random variables.
	
	In expectation we have $\mathbb E(X) = \gamma$. Furthermore, it is unlikely that $X_w$ exceeds twice the expectation by the Chernoff bound from Lemma \ref{lem:chernoffbound}: $\mathbb P \big(X_w \geq 2 \zeta \big) \leq e^{- \zeta/3}$.
	We choose $\zeta \geq 3 (c \ln n + \lambda)$, then $X_w < 2\zeta$ holds w.h.c., by Lemma \ref{lem:whc}. 
	
	This implies that w.h.c., $w$ obtains less than $2\zeta$ requests of the form $\langle \texttt{request}, u\rangle$ from $v$, thus the node $w$ will not filter out $v$'s requests in Lines~\ref{line:filter_out_start}--\ref{line:filter_out_end} and send the response $\langle \texttt{response}, u, B_{uw}\rangle$ to $v$ for each $u \in N_{uw}$. Applying a union bound (Corollary \ref{cor:unionbound_whc}), this holds w.h.c., for each $u,w \in V_h$.
\end{proof}

The next Lemma shows that committees which are not invalid will be made available to all honest nodes in Phase C. Essentially, in the remaining part of Algorithm \ref{alg:PhaseC} starting from Line \ref{line:assemble}, node $v$ looks at all committees $B_{uw}$ it obtained from the nodes $w \in N_{uv}$ and assembles $C_{uv}$ from those nodes that appear in many $B_{uw}$ such that $C_{uv}$ contains the common core $C_u$ if $u$'s committee is valid, else $C_{uv} = \emptyset$. A crucial part of the proof is to show that \smash{$|C_{uv}| < \tfrac{3\beta}{2}$}, for which we exploit that \smash{$|B_{uw}| < \tfrac{5\beta}{4}$} for all honest $w \in N_{uv}$.

\begin{lemma}
	\label{lem:prelim_committees_c}
	Let $t < \tfrac{n}{24}$. We can choose $\zeta \in \bigO(\log n + \lambda)$ such that the following holds for all $u \in V$  w.h.c.	
	\begin{itemize}
		\item If $B_{ux} \neq \emptyset$ for some $x \in V_h$, then $C_u \subseteq C_{uw}$ for all $w \in V_h$,
		\item else $C_{uw} = \emptyset $ for all $w \in V_h$.
		\item \smash{$|C_{uw}| < \tfrac{3\beta}{2}$}, for all $u \in V, w \in V_h$.
	\end{itemize}
\end{lemma}

\begin{proof}
	We condition this proof on the event that all honest nodes correctly exchange the requested committees in Lines \ref{line:request_core} -- \ref{line:respond}, which holds w.h.c.\ for appropriately chosen $\zeta \in \bigO(\log n + \lambda)$ due to Lemma \ref{lem:respond}.
	
	Now fix nodes $v \in V_h$ and $u \in V$. For $w \in V$ let $X_w$ be indicator variable which is 1 if $v$ samples $w$ to $N_{uv}$ in Line \ref{line:sample_core}, else 0. Then \smash{$X_w \sim \text{Ber}(\tfrac{\zeta}{n})$}. %
	Let $Y = \sum_{w \in V_b} X_w$ be the number of Byzantine nodes and $Z = \sum_{w \in V_h, B_{uw} \neq \emptyset} X_w$ the number of honest nodes with nonempty preliminary committee that were sampled to $N_{uv}$. Note that $Y,Z$ are sums of independent random variables.
	
	The expectation of $Y$ is \smash{$\mathbb E (Y) = \tfrac{t\zeta}{n} < \tfrac{\zeta}{24}$}. For the expectation of $Z$ we distinguish two cases, which we describe with an event $E$. 	
	The first case correspons to the event $E$ that there exists an honest node $x \in V_h$ with $B_{ux} \neq \emptyset$, implying that $u$'s committee is valid and in particular \smash{$|S_u| \geq \tfrac{n}{3}$} by Lemma \ref{lem:min_valid_prelim_committees_b}. Thus  \smash{$\mathbb E(Z  \mid E) = \tfrac{|S_u|\zeta}{n} \geq  \tfrac{\zeta}{3}$}. The second case corrseponds the complementary event $\overline E$ where $B_{uw} = \emptyset$ for all $w \in V_h$ (i.e., $u$'s committee is invalid). Clearly $Z = 0$ in the event $\overline E$.

	Let us now analyze the tail probabilities of the random variables $Y,Z$ using the Chernoff bound given in Lemma \ref{lem:chernoffbound}.
	The probability that $Y \geq \zeta/16$ is at most
	\[
	\mathbb P \big(Y \geq \tfrac{\zeta}{16} \big) = \mathbb P \big(Y \geq  (1+\tfrac{1}{2}) \tfrac{\zeta}{24}\big) \leq e^{- d  \zeta},
	\]	
	for some constant $d>0$. The probability that $Z < \zeta/4$ conditioned on case 1 is at most
	\[
		\mathbb P \big(\, Z \leq \tfrac{\zeta}{4} \mid E \, \big) = \mathbb P \big( \, Z \leq  (1 - \tfrac{1}{4})\tfrac{\zeta}{3} \mid E \, \big) \leq e^ {- d  \zeta},
	\]	
	for some constant $d>0$. We choose \smash{$\zeta \geq \tfrac{1}{d}(c\ln n + \lambda)$}, then by Lemma \ref{lem:whc} we have that \smash{$Y < \tfrac{\zeta}{16}$}, w.h.c. For the same reason we have \smash{$Z > \tfrac{\zeta}{4}$} w.h.c.\ if $E$ occurs.  %
	
	In the first case ($E$ occurs), we have \smash{$Z > \tfrac{\zeta}{4}$} w.h.c., which means there is a subset of honest nodes $V' \subseteq V_h$ with $B_{uw} \neq \emptyset$ for all $w \in V'$ and \smash{$|V'| = Z  > \tfrac{\zeta}{4}$}. The overall number $r[u]$ of preliminary committees $B_{uw} \neq \emptyset$ for any $w \in V$ that $v$ obtains is  $r[u] = Y + Z$. 		
	Therefore \smash{$r[u] > \tfrac{\zeta}{4}$} thus the condition in Line \ref{line:invalid_c} does not trigger, and all nodes that are added to $C_{uv}$ are retained.
	
	By Lemma \ref{lem:common_core} it is $C_u \subseteq A_{uw}$ for $w \in V'$, therefore we have $C_u \subseteq B_{uw} $ for all $w \in V'$. In particular, it is \smash{$|V'| = Z = r[u]- Y > r[u]-\tfrac{\zeta}{16}$}. This means that all nodes in $C_u$ satisfy the requirement in Line \ref{line:recover_core} and therefore $C_u \subseteq C_{uv}$. 
	This proves the first property of this lemma.

	In the second case, we have $Z =0$ and since \smash{$Y < \tfrac{\zeta}{16}$} w.h.c., we have \smash{$r[u] = Y + Z <  \tfrac{\zeta}{16}$}, w.h.c. Therefore, the condition in Line \ref{line:invalid_c} triggers and we set $C_{uv} = \emptyset$. This shows the second property of this lemma and also implies the third property of this lemma for case 2.

	It remains to show the third property for case 1, more specifically, we need to show that in case 1 not too many nodes besides the core $C_u$ are added to  $C_{uv}$. (Note that the overall argument for this follows along the lines of Lemma \ref{lem:set_B}, where we constructed the sets $B'_{uv}$ and the size argument of $B'_{uv}$ worked in a similar fashion).

	Every node $x \in C_{uv} \setminus C_u$ that is \emph{not} part of the core $C_u$ but that is still included in $C_{uv}$ must occur in at least \smash{$r[u] - \tfrac{\zeta}{16} - Y$} of the $B_{uw}$ with $w \in V'$ to meet the threshold of \smash{$r[u] - \tfrac{\zeta}{16}$}, let this be property (a). Further, we have \smash{$Z > \tfrac{\zeta}{4}$} honest nodes $w \in V'$ with $C_u \subseteq B_{uw}$ and $|C_u| \geq \beta$ by Lemma \ref{lem:size_committees}, let this be property (b). Further, $|B_{uw}| = |A_{uw}| < \tfrac{5\beta}{4}$ by Lemmas \ref{lem:prelim_committees_b} and  \ref{lem:size_committees}, let this be property (c).
	
	Due to these three properties (a),(b) and (c), the maximum number of nodes in $C_{uv} \setminus C_u$, i.e., nodes that are in $C_{uv}$ \emph{but not} in $C_u$ can be bounded by \smash{$\big(\tfrac{5\beta}{4} \m |C_u|\big) \cdot \tfrac{Z}{r[u]-\zeta/16-Y}$}. The term \smash{$\tfrac{5\beta}{4} \m |C_u|$} describes the maximum number of nodes in some $B_{uw} \setminus C_u$ for honest $w \in V'$ (using properties (b), (c)). We multiply this with $Z$ to obtain an upper bound for the total number of non-core nodes provided by nodes $w \in V'$ via $B_{uw}$. Finally this is divided by $r[u]-\zeta/16-Y$ as this is the lower bound in how many different sets $B_{uw}, w \in V'$ a single non-core node must be contained to be included to $C_{uv}$ (property (a)). We bound this expression as follows.
	\begin{align*}
		\big(\tfrac{5\beta}{4} \m |C_u|\big) \cdot \tfrac{Z}{r[u]-\zeta/16-Y}
		& \leq \big(\tfrac{5\beta}{4} \m |C_u|\big) \cdot \tfrac{r[u]}{r[u]-\zeta/16-Y}\\
		& \hspace*{-3.7mm} \stackrel{Y < \zeta/16}{<} \big(\tfrac{5\beta}{4} \m |C_u|\big) \cdot \tfrac{r[u]}{r[u]-\zeta/8}\\
		& \hspace*{-4.7mm} \stackrel{r[u] > \zeta/4}{\leq} \big(\tfrac{5\beta}{4} \m |C_u|\big)  \cdot \tfrac{\zeta/4}{\zeta/4-\zeta/8}= \tfrac{5\beta}{2}-2|C_u| \stackrel{(b)}{\leq} \tfrac{3\beta}{2} \m |C_u|.
	\end{align*}
	
	Each honest node $x \in C_{uv} \cap V_h$ must be in the common core $C_u$ (by Lemma  \ref{lem:include_honest_nodes_in_committees}). Therefore the overall number of nodes in $C_{uv}$ can be bounded as
	\[
	|C_{uv}|  = |C_{uv} \cap V_h| + |C_{uv} \cap V_b| < |C_u| +  \tfrac{3\beta}{2}-|C_u| = \tfrac{3\beta}{2},
	\]	
	which shows the third property of this lemma.
\end{proof}

Lemma~\ref{lem:prelim_committees_c} shows that for any $u \in V$, the corresponding witness committees $C_{uw}$ are either empty for all honest $w$ (i.e., invalid) or each contains the common core $C_u$ without many additional nodes (i.e., $u$'s committee is valid). We also require that a constant fraction of nodes $u \in V$ have non-empty committees $C_{uw}$ to fulfill the availability property of Definition~\ref{def:witness_committees}. This follows immediately by combining Lemma~\ref{lem:min_valid_prelim_committees_b} with Lemma~\ref{lem:prelim_committees_c}.

\begin{lemma}
	\label{lem:min_valid_prelim_committees_c}
	Let $t < \tfrac{n}{24}$ and $u \in V$. We can choose $\zeta \in \bigO(\log n + \lambda)$ such that the following holds w.h.c.	
	\begin{itemize}
		\item There is a subset $V'$ with $|V'| \geq \tfrac{n}{6}$ such that for all $u \in V'$ it is $C_u \subseteq C_{uw}$ for all $w \in V_h$.
		\item For all $u \in V \setminus V'$ we have $C_{uw} = \emptyset $ for all $w \in V_h$.
		\item \smash{$|C_{uw}| < \tfrac{3\beta}{2}$} for all $u \in V, w \in V\in w_h$.
	\end{itemize}
\end{lemma}

\begin{proof}
	The third property of  Lemma \ref{lem:min_valid_prelim_committees_b} shows that there is a subset $V'$ with $|V'| \geq \tfrac{n}{6}$, such that there exists a $w \in V_h$ with $B_{uw} \neq \emptyset$. Based on this, the first property of Lemma \ref{lem:prelim_committees_c} shows that for each $u \in V'$ we have $C_u \subseteq B_{uw}$ for all honest nodes $w \in V_h$. This gives us the first property of this lemma. The second and third property of this Lemma are immediate consequences of the second and third properties of Lemma \ref{lem:prelim_committees_c}.
\end{proof}

It remains to prove Theorem~\ref{thm:set_up_committee}. We have already shown that the witness committees satisfy the properties in Definition~\ref{def:witness_committees}, largely summarized in Lemma~\ref{lem:min_valid_prelim_committees_c}. The proof now proceeds in two parts: first, we establish the correctness of the overall algorithm (Algorithms \ref{alg:PhaseA}, \ref{alg:PhaseB}, and \ref{alg:PhaseC} executed sequentially) with respect to membership, agreement, and availability. Second, we argue that its round complexity and required bandwidth meet the bounds stated in the theorem.

\begin{proof}[Proof of Theorem \ref{thm:set_up_committee}]
	For the correctness we condensed the main claims we need in Lemma \ref{lem:min_valid_prelim_committees_c}. The following holds w.h.c.:
	
	\begin{itemize}
		\item \textbf{Membership:} $C_{uv}$ contains less than \smash{$\tfrac{3\beta}{2}$} nodes due to the third property of Lemma \ref{lem:min_valid_prelim_committees_c}. 
		An honest node $v \in V_h$ node can only be member of any $C_{uw}$ for $w \in V_h$, if $u \in M_v$ in Line \ref{line:draft_prelim_committee}. Since $v$ adds at most $\gamma$ nodes to $N_{uv}$, $v$ can be member of $C_{uw}$ for $w \in V_h$, for at most $\gamma$ distinct $u \in V$. Due to the value $\gamma$ chosen in Lemma \ref{lem:common_core_size}, we have $\gamma \leq 2 \beta$.
		
		\item \textbf{Agreement:} Lemma \ref{lem:min_valid_prelim_committees_c} gives 
		us the property that for all $u \in V$ it is either $C_u \subseteq C_{uw}$ for all $w \in V_h$ or $C_{uw} = \emptyset $ for all $w \in V_h$. This gives us agreement either on validity of $u$'s committee and specifically on the common core $C_u$ or the invalidity of $u$'s committee. The minimum size of $\beta$ of the core $C_u$ was shown in Lemma \ref{lem:common_core_size}.
		
		\item \textbf{Availability:} In particular, Lemma \ref{lem:min_valid_prelim_committees_c} shows that we have $C_u \subseteq C_{uw}$ for all $w \in V_h$ for at least $\tfrac{n}{6}$ nodes $u \in V$, which implies the availability property.
	\end{itemize}	

We now analyze the round complexity and per-node bandwidth usage. Our pseudocode separates local computations, which complete instantly in this model, from message exchanges that appear under ``\textbf{send}~\dots~\textbf{to}~\dots~\textbf{foreach}~\dots'' and require multiple rounds. One exception is the single call to a consensus subroutine in Algorithm~\ref{alg:PhaseB}, where nodes in the core $C_u$ exchange $\bigO(\beta^5) = \bigO(\poly{\log n + \lambda})$ bits per node (Lemma~\ref{lem:flag_b_2}).

For other transmissions, each honest node sends at most $n \cdot x$ total messages, where $x \in \bigO(\log n + \lambda)$. This follows because $x$ is proportional to $\beta, \gamma, \zeta$, each bounded by $\bigO(\log n + \lambda)$ (Lemmas~\ref{lem:size_committees}, \ref{lem:respond} and \ref{lem:prelim_committees_c}). Each message contains at most $\beta$ node identifiers, requiring $\beta \log n \in \bigO\big(\log n(\log n + \lambda)\big)$ bits. Multiplying by the total number of messages, we obtain an overall bit complexity of $\bigO\big(n \log n (\log n + \lambda)^2\big) = \bigO\big(n \cdot \poly{\log n + \lambda}\big)$.

The required number of rounds is $\bigO\big(\tfrac{n \cdot \poly{\log n + \lambda}}{\sigma}\big)$, where $\sigma$ is the bandwidth (bits per round) available to each node. While we can trade off between round complexity and message size through $\sigma$, the theorem's claim holds for $\sigma = \poly{\log n + \lambda}$.
\end{proof}

\section{Solving Distributed Tasks with Witness Committees}
\label{sec:applications}
	
In the second part of this article we aim to implement solutions to distributed problems using the witness committees that we computed in the previous section. Specifically, witness committees allow us to tackle the generic tasks of disseminating and aggregating information as well as forming consensus with resiliency to Byzantine failures. 

The road-map of this section is that we will first introduce specialized problems for which we develop partial solutions in conjunction with the system of witness committees. These will then serve as subroutines in algorithms for the aforementioned tasks.

Note that most of the algorithms in this section are deterministic and designed to work in an asynchronous setting. Furthermore, our algorithms build on top of each other as a protocol stack. For these reasons we present all algorithms as event-based pseudo-code where each algorithm implements a solution to a specific distributed problem we carefully define in advance (style and syntax adhere closely to \cite{DBLP:books/daglib/0025983}).

\subsection{Lazy Consensus within Witness Committees}

Let us first consider consider a sub-problem we denote with ``lazy consensus,'' used as a subroutine for reliable broadcast. Unlike the consensus in Definition~\ref{def:consensus}, nodes in a subset $V' \subseteq V$ must decide a value only if all honest nodes in $V'$ propose it; otherwise, no decision is required. However, if one honest node in $V'$ decides a value, every other honest node in $V'$ must decide the same value.

\begin{definition}[Lazy Consensus]
	\label{def:lazy_consensus}
	Let $V' \subseteq V$ and each $v \in V'$ proposes a value. The lazy consensus problem is solved when the following holds.
	\begin{itemize}
		\item \textbf{Integrity}: Each honest $v \!\in\! V'$ decides at most one value. If $\,V'\!$ contains honest nodes, any $v \in V'$ only decides a value proposed by an honest node in $V'$.
		\item \textbf{Validity}: If all honest $v \in V'$ propose the same value, then every node decides that value.
		\item \textbf{Agreement}: If an honest node $v \in V'$ decides a value, then every honest node in $V'$ decides the same.
	\end{itemize}	
\end{definition}

This problem cannot be solved on any subset $V'$ of $V$, first because nodes in $V'$ might not know of each other, second because $V'$ could contain a large number of Byzantine nodes. Instead we give an algorithm that solves this problem on one of the witness committees $C_u$ for $u \in V$ in. 	
The key observation is that Bracha's algorithm \cite{DBLP:conf/podc/Bracha84}  essentially contains a solution of the case where $V' = V$ and less than $n/3$ nodes are Byzantine. We employ an adaptation of Bracha's algorithm configured to work within a witness committee.

\begin{algorithm}
	\caption{\textbf{Lazy Consensus} \Comment{\textit{executed by $v \in V$}}}
	\label{alg:lc}
	
	\SetKwInOut{Input}{input}
	\SetKwInOut{Output}{output}	
	\SetKwInOut{Implements}{implements}
	\SetKwInOut{Requires}{precondition}	
	\SetKwProg{On}{upon}{ do}{end}
	\SetKwProg{OnEvent}{upon event}{ do}{end}
	\SetKw{Or}{or}
	\SetKw{And}{and}
	
	\DontPrintSemicolon		
	
	\Implements{Instance \textit{lc} that solves the lazy consensus problem for $V' = C_u$ for some $u$ according to Definition \ref{def:lazy_consensus}.}
	\Requires{System of witness committees satisfying Definition \ref{def:witness_committees}}
	
	\vspace*{0.3\baselineskip}		
	
	\OnEvent{$\langle \textit{lc}, \texttt{propose}, x , u \rangle$ \And $C_{uv} \neq \emptyset$ \And $v \in C_{uv}$}{send $\langle \texttt{echo}, x, u \rangle$ to all $r \in C_{uv}$ \Comment{\textit{$v$ also ``sends'' the echo to itself}}}

	\vspace*{0.2\baselineskip}	
	
	\On{receiving $\langle \texttt{echo}, x, u \rangle$ from at least $\beta$ distinct nodes $w \in C_{uv}$ \And $v \in C_{uv}$\label{line:vote1}}{send $\langle \texttt{vote}, x, u \rangle$ to all $w \in C_{uv}$}
	
	\vspace*{0.2\baselineskip}	
	
	\On{receiving $\langle \texttt{vote}, x, u \rangle$ from at least $\beta/2$ distinct nodes $w \in C_{uv}$ \And $v \in C_{uv}$ \label{line:vote2}}{send $\langle \texttt{vote}, x, u \rangle$ to all $w \in C_{uv}$}
	
	\vspace*{0.2\baselineskip}	
	
	\On{receiving $\langle \texttt{vote}, x, u \rangle$ from at least $\beta$ distinct nodes $w \in C_{uv}$ \And $v \in C_{uv}$ \label{line:lc_decide}}{\textbf{trigger} $\langle \textit{lc}, \texttt{decide}, x \rangle$}		
\end{algorithm}

\begin{lemma}
	\label{lem:lazy_consensus}
	Let $u \in V$ with $C_{uv} \neq \emptyset$ for $v \in V_h$. %
	Then Algorithm \ref{alg:lc} solves the lazy consensus problem as given in Definition \ref{def:lazy_consensus} for $V' = C_u$. The number of messages sent by each honest node is $\bigO(\beta)$.
\end{lemma}

\begin{proof}
	\textbf{Integrity:} for an honest node $v \in C_u$ to decide $x$ in Line \ref{line:lc_decide}, it requires messages $\langle \texttt{vote}, x, u \rangle$ from at least $\beta$ distinct nodes in $C_{uv}$. Note that, since $v$ knows a ``local'' version $C_{uv}$ of $u$'s committee, it can deduce from whom to accept these messages. Since \smash{$|C_{uv}| < \tfrac{3\beta}{2}$} this condition can be satisfied for at most one value $x$. 		
	
	For any honest $v \in C_u$ to decide on a value $x$ in Line \ref{line:lc_decide}, at least $\beta$ votes $\langle \texttt{vote}, x, u \rangle$ for $x$ are required. Since $|C_u| \geq \beta$ and  \smash{$|C_{uv}| < \tfrac{3\beta}{2}$},  this can only happen if $x$ was voted for by at least $\beta/2$ nodes in $C_u$. There are two possibilities that $v \in C_u$ sends $\langle \texttt{vote}, x, u \rangle$ in the first place, either in Line \ref{line:vote1} or in Line \ref{line:vote2}.
	
	Assume that $v$ sends a vote for $x$ in Line  \ref{line:vote1}. That means that $x$ must have been echoed by at least $\beta$ nodes in $C_{uv}$.This can only happen for a value $x$ that was proposed by at least $\beta/2$ honest nodes. The other possibility is that $v$ sends a $\langle \texttt{vote}, x, u \rangle$ after fulfilling the condition in Line \ref{line:vote2}. Note that this cannot be caused by the Byzantine nodes voting for some value $x$. Consequently, $x$ must have been voted for by at least one honest node after satisfying \ref{line:vote1}, which by the argument above implies that $x$ was proposed by at least $\beta/2$ honest nodes.
	
	\textbf{Validity:} Assume all nodes in $C_u$ propose $x$. Then each $w \in C_u$ obtains a message $\langle \texttt{echo}, x, u \rangle$ from every node $C_u$ (including an echo from $v$ to itself). Since $|C_u| \geq \beta$ this implies that each node in $v \in C_u$ sends a message $\langle \texttt{vote}, x, u \rangle$ to $C_{uv}$, thus each node in $C_u$ obtains at least $\beta$ messages $\langle \texttt{vote}, x, u \rangle$ which satisfies Line \ref{line:lc_decide} thus $v$ decides $x$.
	
	\textbf{Agreement:} we first show that for any two nodes $v_1,v_2 \in C_u$ that send $\langle \texttt{vote}, x_1, u \rangle$, $\langle \texttt{vote}, x_2, u \rangle$ in Line \ref{line:vote1} it is $x_1 = x_2$. For a contradiction assume $x_1 \neq x_2$. This implies that there is a subset $V_1 \subseteq C_{uv_1}$ with $|V_1| \geq \beta$ that all sent $\langle \texttt{echo}, x_1, u \rangle$ to $v_1$. Analogously  for $v_2$, there is a corresponding subset $V_2 \subseteq C_{uv_2}$ with $|V_2| \geq \beta$ that echoed $x_2$.

	Since $C_u \subseteq C_{uv_1}$,  \smash{$|C_{uv_1}|< \tfrac{3\beta}{2}$} and $|C_u| \geq \beta$, we have \smash{$|C_{uv_1} \setminus C_u| < \tfrac{\beta}{2}$}, therefore \smash{$|V_1\setminus C_u| < \tfrac{\beta}{2}$}. 
	This implies \smash{$|V_1 \cap C_u| = |V_1| \m |V_1\setminus C_u| > \beta \m \tfrac{\beta}{2} = \tfrac{\beta}{2}$}. The same for $v_2$:  \smash{$|V_2 \cap C_u| > \tfrac{\beta}{2}$}. This further implies that $V_1$ and $V_2$ intersect in at least one node $w \in C_u$, i.e., $w$ must have sent $\langle \texttt{echo}, x_1,u \rangle$ to $v_1$ and $\langle \texttt{echo}, x_2,u \rangle$ to $v_2$, a contradiction as $w \in C_u$ is honest.
	
	Now let us show that the same holds for the messages $\langle \texttt{vote}, x, u \rangle$ that nodes in $C_u$ send in Line \ref{line:vote2}, i.e., that $x$ is the same for all nodes in $C_u$. Let $v$ be the \emph{first} node that sends $\langle \texttt{vote}, x, u \rangle$ in Line \ref{line:vote2}. Let $V^* \subseteq C_{uv}$ be the set of at least  $\beta/2$ nodes from which $v$ received $\langle \texttt{vote}, x, u \rangle$. Since \smash{$|V^*\setminus C_u| < \tfrac{\beta}{2}$}, the set $V^*$ must include a node $w \in C_u$.
	
	As $v$ is the first to sent $\langle \texttt{vote}, x, u \rangle$ in Line \ref{line:vote2}, $w$ must have sent its message $\langle \texttt{vote}, x, u \rangle$ in Line \ref{line:vote1}. Therefore $v$ must have voted for the same value $x$ that $w$ and any other node in $C_u$ voted for in Line \ref{line:vote1}. Repeating this argument inductively, we obtain that $x$ must also be the same for all messages $\langle \texttt{vote}, x, u \rangle$ that are sent by nodes in $C_u$ in Line \ref{line:vote2}.
	
	Knowing that all $v \in C_u$ that sent $\langle \texttt{vote}, x, u \rangle$ have the same value $x$ we can show agreement as follows. Assume that $v$ decides $x$ in Line \ref{line:bc_n2c_deliver}. This means it received $\beta$ votes $\langle \texttt{vote}, x, u \rangle$  from nodes in $C_{uv}$. Since $|C_u| \geq \beta$ and $|C_{uv}|< 3\beta/2$, node $v$ received more than $\beta/2$ votes from $C_u$. 
	
	This means that these $\beta/2$ votes $\langle \texttt{vote}, x, u \rangle$ will reach every other node in $C_u$ so each node in $C_u$ will themselves send $\langle \texttt{vote}, x, u \rangle$ in Line \ref{line:vote2} all with the same value $x$. Since $|C_u| \geq \beta$ this implies that the condition in Line \ref{line:lc_decide} is met by every node in $C_u$ eventually, so every node in $C_u$ will decide $x$ eventually.
	
	\textbf{Number of messages:} Each node $v$ broadcasts to $C_{uv}$ at most 4 times and $|C_{uv}| \in \bigO(\beta)$.
\end{proof}

\subsection{Reliable Broadcast among Witness Committees}

This subsection aims to implement reliable broadcast among specific witness committees,
rather than the full network. To that end, we define a generalized version of reliable broadcast
for sets of senders and receivers: each sender has a message to transmit, and each receiver
needs to deliver exactly one message. All honest receivers must deliver the same message,
and if all honest senders agree on a message, that same message must be delivered. If there
are no honest senders at all, then no delivery is required.

\begin{definition}[Generalized Reliable Broadcast]
	\label{def:generalized_reliable_broadcast}
	Let $S \subseteq V$ be a set of sender nodes, where each honest node $s \in S$ has at most one message it wants to send.\footnote{For simplicity we consider only the number of messages sent by each node assuming neglecting the size of the broadcast message. The load per node in bits is proportional to the number of messages times the size of the original broadcast message.} Let $R \subseteq V$ be a dedicated set of receiver nodes, not necessarily disjoint with $S$. The reliable broadcast problem with senders $S$ and receivers $R$ is solved if the following properties hold.
	\begin{itemize}
		\item \textbf{Integrity}: Every honest $r \in R$ delivers at most one message. If $S\!$ contains honest nodes, then any delivered message must have been sent by some honest $s \in S$.
		\item \textbf{Validity}: If all honest $s \in S$ have the same msg.\ $M$, then each honest $r \in R$ delivers $M$.
		\item \textbf{Agreement}: If any honest $r \in R$ delivers $M$, then all honest receivers in $R$ deliver $M$.
	\end{itemize}
\end{definition}
	
	Note that the lazy consensus defined earlier can be seen as a special case of the generalized reliable broadcast for $S = R$. Furthermore, the standard reliable broadcast problem (Defintion \ref{def:reliable_broadcast}) corresponds to the case where there is just one sender and $V$ is the set of receivers. %

	As in the lazy consensus problem, the nodes in $S$ and $R$ might not necessarily know of each other and might contain Byzantine majorities, thus we cannot hope to solve the generalized reliable broadcast problem from Definition \ref{def:generalized_reliable_broadcast} in the Byzantine setting {for arbitrary subsets}. However, in the following we show that generalized reliable broadcasts within a system of witness committee can actually be done efficiently, by exploiting its structure and the common knowledge of it.

	Using Algorithm \ref{alg:lc} as subroutine, we implement a reliable broadcast from a dedicated source node $s$ to the committee of some node $u$ in Algorithm \ref{alg:bc_node_committee}. This requires that $u$'s committee it is valid, i.e. every node $v \in V_h$ knows a witness committee $C_{uv} \neq \emptyset$ with the properties from Definition \ref{def:witness_committees}. 
	In our pseudo code, Algorithm \ref{alg:bc_node_committee} represents an implementation of a generic module \textit{rbc-n2c} that supplies a solution of the reliable broadcast problem from Definition \ref{def:generalized_reliable_broadcast} from a single node $S = \{s\}$ to a committee $R = C_u$.

\begin{algorithm}
	\caption{\textbf{Reliable Broadcast - Node to Committee} \Comment{\textit{executed by $v \in V$}}}
	\label{alg:bc_node_committee}
	
	\SetKwInOut{Input}{input}
	\SetKwInOut{Output}{output}	
	\SetKwInOut{Implements}{implements}
	\SetKwInOut{Requires}{precondition}	
	\SetKwProg{On}{upon}{ do}{end}
	\SetKwProg{OnEvent}{upon event}{ do}{end}
	\SetKw{Or}{or}
	\SetKw{And}{and}
	
	\DontPrintSemicolon		
	
	\Implements{Reliable broadcast instance \textit{rbc-n2c} from node $s$ to $u$'s committee with $S = \{s\}$ and $R = C_u$ according to Definition \ref{def:generalized_reliable_broadcast}.}
	\Requires{System of witness committees satisfying Definition \ref{def:witness_committees}}
	
	\vspace*{0.3\baselineskip}		
	
	\OnEvent{$\langle \textit{rbc-n2c}, \texttt{broadcast}, M, u \rangle$ \And $C_{uv} \neq \emptyset$}{send $\langle \texttt{transmit-n2c}, M, v, u \rangle$ to all $r \in C_{uv}$}
	
	\vspace*{0.2\baselineskip}	
	
	\On{receiving $\langle \texttt{transmit-n2c}, M, s, u \rangle$ from $s$ \And $v \in C_{uv}$}{\textbf{trigger} $\langle \textit{lc}, \texttt{propose}, M, u \rangle$\Comment{\textit{see Algorithm \ref{alg:lc}}}}
		
	\vspace*{0.2\baselineskip}	
	
	\OnEvent{$\langle \textit{lc}, \texttt{decide}, M \rangle$ \label{line:bc_n2c_deliver}}{\textbf{trigger} $\langle \textit{rbc-n2c}, \texttt{deliver}, M \rangle$}		
\end{algorithm}

To obtain correctness and performance guarantees of Algorithm \ref{alg:bc_node_committee} we exploit the properties of Algorithm \ref{alg:lc} for lazy consensus.

\begin{lemma}
	\label{lem:bc_node_committee}
	Let $u \in V$ with $C_{uv} \neq \emptyset$ for $v \in V_h$.  %
	Algorithm \ref{alg:bc_node_committee} solves the reliable broadcast problem as given in Definition \ref{def:generalized_reliable_broadcast} between a dedicated source node $s \in V$ and the common core $C_u$. The number of messages sent by the honest nodes from the set $C_u \cup \{s\}$ is $\bigO(\beta)$.
\end{lemma}

\begin{proof}
	\textbf{Integrity:} for honest node $v \in C_u$ to deliver $M$ in Line \ref{line:bc_n2c_deliver}, there must have been a decision on $M$ using the lazy consensus. By Lemma \ref{lem:lazy_consensus}, Algorithm \ref{alg:lc} solves the lazy consensus problem, thus integrity of this broadcast follows from the integrity of lazy consensus.
	
	\textbf{Validity:} we show that if $s$ is honest, then all $r \in C_u$ deliver $M$. Since $C_u \subseteq C_{us}$ and all nodes in $C_u$ are honest, each $r \in C_u$ obtains a message $\langle \texttt{transmit-n2c}, M, s, u \rangle$ from $s$. Therefore, each $r \in C_u$ proposes the same message $M$ and by the validity property of lazy broadcast also decides on $M$. Consequently,  each $r \in C_u$ triggers the condition on Line \ref{line:bc_n2c_deliver} for message $M$ and delivers $M$.
	
	\textbf{Agreement:} If some message $M$ is delivered by some $v \in C_u$, then $v$ decided $M$ in the lazy consensus as per the condition in Line \ref{line:bc_n2c_deliver}. By the agreement condition of lazy consensus, every other node in $C_u$ will decide $M$ as well, which implies that every node in $C_u$ will deliver $M$ eventually.
	
	\textbf{Number of messages:} By Lemma \ref{lem:lazy_consensus} at most $\bigO(\beta)$ per honest node are required for solving the lazy broadcast instance and the sender broadcasts once to $C_{uv}$ and $|C_{uv}| \in \bigO(\beta)$.
\end{proof}

Next, we give a broadcast from a common core of a valid committee to a single node $r \in V$. Algorithm \ref{alg:bc_committee_node} implements a reliable broadcast module \textit{rbc-c2n} that solves the reliable broadcast problem from Definition \ref{def:generalized_reliable_broadcast} with $S=C_u$ and $R = \{ r \}$. The basic idea is that a message $M$ is only delivered by $r$ if $M$ was received by at least $\beta$ distinct nodes in $C_{ur}$.

\begin{algorithm}
	\caption{\textbf{Reliable Broadcast - Committee to Node} \Comment{\textit{executed by $v \in V$}}}
	\label{alg:bc_committee_node}
	
	\SetKwInOut{Input}{precondition}
	\SetKwInOut{Output}{output}	
	\SetKwInOut{Implements}{implements}
	\SetKwProg{On}{on}{ do}{end}
	\SetKwProg{OnEvent}{on event}{ do}{end}	
	\SetKw{Or}{or}
	\SetKw{And}{and}
	
	\DontPrintSemicolon		
	
	\Implements{Reliable Broadcast instance \textit{rbc-c2n} from $S = C_u$ to a single node $R=\{r\}$ according to Definition \ref{def:generalized_reliable_broadcast}.}
	\Input{System of witness committees satisfying Definition \ref{def:witness_committees}}
	
	\vspace*{0.3\baselineskip}	
	
	\OnEvent{$\langle \textit{rbc-c2n}, \texttt{broadcast}, M, u, r \rangle$ \And $v \in C_{uv}$ \And $C_{ur} \neq \emptyset$}{send $\langle \texttt{transmit-c2n}, M, u, r \rangle$ to $r$}
	
	\vspace*{0.2\baselineskip}	
	
	\On{receiving $\langle \texttt{transmit-c2n}, M, u, r\rangle$ from $\geq \beta$ distinct nodes $w \in C_{ur}$ \And $v = r$\label{line:bc_c2n_deliver}}{\textbf{trigger} $\langle  \textit{rbc-c2n}, \texttt{deliver}, M\rangle$}	
\end{algorithm}

\begin{lemma}
	\label{lem:bc_committee_node}
	Let $r,u \in V$ with $C_{uv} \neq \emptyset$ for $v \in V_h$. 
	Then Algorithm \ref{alg:bc_committee_node} solves the reliable broadcast problem as given in Definition \ref{def:generalized_reliable_broadcast} between the common core $S = C_u$ and $R = \{r\}$. Each honest node in $C_u$ sends at most one message.
\end{lemma}

\begin{proof}
	\textbf{Integrity:} 
	Since $ |C_{ur}|  < \tfrac{3\beta}{2}$ the threshold of $\beta$ messages $\langle \texttt{transmit-c2n}, M, u, r\rangle$ in Line \ref{line:bc_c2n_deliver} can be met by at most one message $M$. Further, since $ |C_{u}|  \geq \beta$ any message $M$ that was delivered was sent by at least one honest node in $C_u$.

	\textbf{Validity:} 	Assume that all $v \in C_u$ want to send the same message $M$ to $r$. Each $v \in C_u$ sends the message $\langle \texttt{transmit-c2n}, M, u, r\rangle$ to $r$. Since $|C_u| \geq \beta$, node $r$ will meet the threshold in Line \ref{line:bc_c2n_deliver} and deliver $M$.	 
	
	\textbf{Agreement}: This is clear since the receiver is a single node.
	
	\textbf{Number of Messages:} Each node in $C_u$  sends a message to $C_{r}$.
\end{proof}

Finally, we implement a reliable broadcast  between two witness committees of nodes $s,r$, given that both have valid committees $\neq \emptyset$. More specifically, Algorithm \ref{alg:bc_committees} implements a module \textit{rbc-c2c} that can be used to solve the reliable broadcast problem from Definition \ref{def:generalized_reliable_broadcast} between $S = C_s$ and $R = C_r$. 

The reliable broadcast between committees works in two steps. Roughly speaking, first we use the reliable broadcast from $C_s$ to single nodes to inform each  $w \in C_r$ about the message $M$ that $C_s$ wants to communicate. Subsequently we use the lazy consensus protocol among nodes in $C_r$ to get a consensus on the messages that were received by the nodes in $C_r$, since the latter might have been subject to Byzantine influence.

\begin{algorithm}
	\caption{\textbf{Reliable Broadcast - Committee to Committee} \Comment{\textit{executed by $v \in V$}}}
	\label{alg:bc_committees}	

	\SetKwInOut{Input}{input}
	\SetKwInOut{Output}{output}	
	\SetKwInOut{Implements}{implements}
	\SetKwInOut{Requires}{precondition}	
	\SetKwProg{On}{upon}{ do}{end}
	\SetKwProg{OnEvent}{upon event}{ do}{end}
	\SetKw{Or}{or}
	\SetKw{And}{and}
	
	\DontPrintSemicolon		
	
	\Implements{Reliable broadcast instance \textit{rbc-c2c} among valid witness committees with $S = C_{s}, R = C_{r}$ according to Definition \ref{def:generalized_reliable_broadcast}.}  %
	\Requires{System of witness committees satisfying Definition \ref{def:witness_committees}}

	\vspace*{0.3\baselineskip}		
	
	\OnEvent{$\langle \textit{rbc-c2c}, \texttt{broadcast}, M, s, r \rangle$ \And $v \in C_{s v}$ \And $C_{rv} \neq \emptyset$}
	{
		\ForEach{$w \in C_{rv}$}
		{
			\textbf{trigger} $\langle \textit{rbc-c2n}, \texttt{broadcast}, [M,s,r], s, w \rangle$ \Comment{\textit{see Algorithm \ref{alg:bc_committee_node}}}
		}
	}
	
	\vspace*{0.1\baselineskip}	
	
	\OnEvent{$\langle \textit{rbc-c2n}, \texttt{deliver}, [M, s,r] \rangle$ \And $v \in C_{rv}$}{\textbf{trigger} $\langle \textit{lc}, \texttt{propose}, [M, s,r]\rangle$\Comment{\textit{see Algorithm \ref{alg:lc}}}}
	
	\vspace*{0.2\baselineskip}	
	
	\OnEvent{$\langle \textit{lc}, \texttt{decide}, [M, s,r] \rangle$ \And $v \in C_{r v}$\label{line:bc_c2c_deliver}}{\textbf{trigger} $\langle \textit{rbc-c2c}, \texttt{deliver}, M \rangle$}
\end{algorithm}

\begin{lemma}
	\label{lem:bc_committees}
	Let $s,r \in V$ with $C_{sv},C_{rv} \neq \emptyset$ for $v \in V_h$. %
	Then Algorithm \ref{alg:bc_committees} solves the reliable broadcast problem as given in Definition \ref{def:generalized_reliable_broadcast} between the common cores $C_s$ and $C_r$. The number of messages sent by each  node in $C_s$ and $C_r$ is $\bigO(\beta)$.
\end{lemma}

\begin{proof}
	
	\textbf{Integrity:} We first conduct a broadcast from $C_s$ to  each $w \in C_r$ using Algorithm \ref{alg:bc_committee_node}. By Lemma \ref{lem:bc_committee_node} each node $w \in C_r$ that delivers a message, will deliver one that was sent by at least one honest node in $C_s$. In the second step we conduct a lazy consensus on the received messages. By the integrity property of lazy consensus shown in Lemma \ref{lem:lazy_consensus} we will decide at most one message that was proposed by an honest node. 
				
	\textbf{Validity:} 
	Assume that all $v \in C_s$ want to send the same message $M$ to $C_r$. By the validity of broadcasting from committee to nodes shown in Lemma \ref{lem:bc_committee_node}, each node in $C_r$ delivers $M$ in the first step after completing Algorithm \ref{alg:bc_committee_node}. By the validity property of lazy consensus shown in Lemma \ref{lem:lazy_consensus}, each node will then decide $M$ in the second step after Algorithm \ref{alg:lc} is complete.
		
	\textbf{Agreement}: Let $M$ be the message that is delivered by some $v \in C_r$ after the condition in Line \ref{line:bc_c2c_deliver} is satisfied. By the agreement property of  lazy consensus each node in $C_r$ must have decided $M$ and thus delivered $M$ as well.

	\textbf{Number of Messages:} Each node in $v \in C_s$  sends a message to each node in $C_{rv}$ and $|C_{rv}| \in \bigO(\beta)$ by Lemma \ref{lem:bc_committee_node}. Subsequently, the lazy broadcast incurs another $\beta$ messages per node in $C_r$ by Lemma \ref{lem:lazy_consensus}.	
\end{proof}

\textbf{Concurrent Broadcast Sessions.} All of our broadcast algorithms are described as single-shot procedures, assuming no concurrent broadcasts are ongoing---whether initiated intentionally or by the adversary. If concurrency arises, we can distinguish parallel broadcast sessions by enriching messages with unique sender and receiver identifiers and a broadcast instance ID. Honest nodes then verify that the sender and receiver information match the intended nodes or committees and that a broadcast is indeed expected (e.g., they are neighbors in the broadcast tree introduced in the next section). This prevents the adversary from injecting spurious messages or hijacking legitimate sessions, ensuring that unintended broadcasts are filtered out.

\subsection{A Broadcast Tree for Witness Committees}

Previously, we established efficient communication among witness committees. It remains to coordinate such communication on a network wide level, in order to implement a reliable broadcast efficiently with near constant time and workload per node. In particular we want to compute the following tree structure for communication among witness committees. 

\begin{definition}[Broadcast Tree for Witness Committees] 	
	\label{def:broadcast_tree}
	Let $C_{uv}$ be a system of witness committees as in Definition \ref{def:witness_committees} with valid committees $C_{uv} \neq \emptyset$ for all $v \in V$ for a constant fraction of at least $\alpha n$ nodes $u \in V$. A \emph{broadcast tree} $T$ with degree $\delta \in \mathbb N, \delta \geq \tfrac{1}{\alpha} \p 1$ for the witness committees $C_{uv}$ satisfies the following properties:
	\begin{itemize}
		\item Each node in $T$ has at most $\delta$ child nodes and height $\bigO(\log_{\delta} n)$.
		\item Each inner node corresponds to a node $u$ with a valid committee.
		\item Each node $u$ with a valid committee corresponds to at most one inner node.
		\item Each node $v \in V$ corresponds to a leaf in the tree.
		\item Every honest node knows the tree $T$ including the root committee \texttt{root}$(T)$. 		
	\end{itemize}
	
\end{definition}

Each node can compute such a broadcast tree locally, such that the tree is consistent among honest nodes. For this, we crucially exploit that there already is agreement on the set of nodes $u \in V$ that have a valid committee $C_{uv} \neq \emptyset$. We obtain the following lemma.

\begin{lemma}
	\label{lem:broadcast_tree}
	Assume we have witness sets as in Definition \ref{def:witness_committees}. Then a broadcast tree $T$ with degree $\delta \geq \tfrac{1}{\alpha} \p 1$ as in Definition \ref{def:broadcast_tree} can be computed by all nodes locally.
\end{lemma}

\begin{proof}
	Let $\delta \geq \tfrac{1}{\alpha} \p 1$. Each honest node $v \in V_h$ constructs the set of inner nodes as follows. Let $T$ be a perfect tree of degree $\delta$ and height $h:= \lceil \log_{\delta} n \rceil$, i.e., it has at least $n$ leaves.
	Node $v$ now sorts the $u \in V$ with $C_{uv} \neq \emptyset$ by identifier into a sequence $(u_1, \ldots, u_k)$. Then $v$ assigns the nodes $u_i$ to inner nodes in $T$ in the order of this sequence, layer by layer and from left to right. Note that not all nodes $u_i$ might be needed to fill the inner nodes of $T$. 
	
	However, since we have $\alpha n$ nodes with valid committees and $\delta \geq \tfrac{1}{\alpha} + 1$, we have enough nodes with valid committees to assign to the inner nodes of $T$, since the  number of inner nodes of $T$ is at most \smash{$\tfrac{\delta^h-1}{\delta-1} \leq \tfrac{n}{\delta -1} \leq \alpha n$}.	
	Then $v$ sorts all nodes in $V$ by identifier into a sequence $(v_1, \ldots, v_n)$ and assigns them to the leaves of $T$ in this order. Since $T$ has at least $n$ leaves, each node in $V$ will be assigned a position as leaf. 
	
	Finally, note that $T$ depends only on the sorted sequences  $(u_1, \ldots, u_k)$ and $(v_1, \ldots, v_n)$. Since there is agreement among all honest nodes which nodes $u \in V$ have valid committees, the sequences of nodes that $T$ depends on are the same for all honest nodes, thus each honest node will compute the same tree $T$.
\end{proof}

\subsection{Reliable Broadcast}

We now have the ingredients to efficiently implement a reliable broadcast as given in Definition \ref{def:reliable_broadcast}, with our goal to keep the round complexity and overall message complexity \textit{per node} near-constant and without using signatures.
Further, our approach does not require the use of signatures, which are essentially substituted by using the broadcast tree among witness committees instead. Finally, all requirements regarding synchrony and randomization are effectively shifted to the computation of the system of witness committees, thus our reliable broadcast is deterministic and works in asynchronous networks.

The following algorithm implements a module \textit{rbc} that solves reliable broadcast. It uses the subroutines for the generalized reliable broadcast among single nodes and witness committees defined in Algorithms \ref{alg:bc_node_committee}, \ref{alg:bc_committee_node} and \ref{alg:bc_committees}. It also relies on the fact that the broadcast tree $T$ for the witness committee structure can be computed locally (Lemma \ref{lem:broadcast_tree} and Definition \ref{def:witness_committees}, respectively). 

The idea is not complicated given that we have subroutines to communicate reliably among single nodes and the witness committees within the broadcast tree $T$. First we reliably broadcast the message from the sender to its parent committee in $T$. The message is then handed from committee to committee up the tree to the root committee. From the root committee we broadcast the message down to all leaves, i.e., all nodes in $V$. While the pseudocode may look cumbersome, most lines are caused by the necessary case distinction for root and leaves in $T$.

\begin{algorithm}
	\caption{\textbf{Reliable Broadcast} \Comment{\textit{executed by $v \in V$}}}
	\label{alg:bc}
		
	\SetKwInOut{Input}{precondition}
	\SetKwInOut{Output}{output}	
	\SetKwInOut{Implements}{implements}
	\SetKwInOut{Requires}{requires}	
	\SetKwProg{On}{upon}{ do}{end}
	\SetKwProg{OnEvent}{upon event}{ do}{end}
	\SetKw{Or}{or}
	\SetKw{And}{and}
	
	\DontPrintSemicolon		
	
	\Implements{Reliable broadcast instance \textit{rbc} that solves the reliable broadcast problem according to Definition \ref{def:reliable_broadcast}.}
	\Input{System of witness committees satisfying Definition \ref{def:witness_committees}. Broadcast tree $T$ with degree $\delta$ as described in Definition \ref{def:broadcast_tree}.}
	
	\vspace*{0.3\baselineskip}	
	
	\OnEvent{$\langle \textit{rbc}, \texttt{broadcast}, M \rangle$}{		
		$u \gets \text{ parent of } v \text{ in } T$ \Comment{\textit{$u$ is the parent committee $C_{uv}$ of $v$ in $T$}}	
		
		\textbf{trigger} $\langle \textit{rbc-n2c}, \texttt{broadcast}, [\texttt{up},M,u], u \rangle$ \Comment{\textit{see Algorithm \ref{alg:bc_node_committee}}}
		
	}

	\vspace*{0.3\baselineskip}		
	
	\OnEvent(\Comment{\textit{encloses case distinction for root}}){$\langle \textit{rbc}, \texttt{cast-up}, M, u \rangle$ }{		
		$r \gets \texttt{root}(T)$
		
		\If{$u \neq r$}{
			
				$p \gets \text{ parent of } u \text{ in } T$ 
			
				\textbf{trigger} $\langle \textit{rbc-c2c}, \texttt{broadcast}, [\texttt{up}, M, p], u, p \rangle$ \Comment{\textit{see Algorithm \ref{alg:bc_committees}}}	
		}

		\If{$u = r$}{			
			\textbf{trigger} $\langle \textit{rbc}, \texttt{cast-down}, M, u \rangle$
		}
	}		

	\vspace*{0.3\baselineskip}				
	
	\OnEvent(\Comment{\textit{encloses case distinction for leaves}}){$\langle \textit{rbc}, \texttt{cast-down}, M, u \rangle$}{
		\ForEach{child $c$ of $u$ in $T$}{
			\If{$c$ is not a leaf in $T$}{
				\textbf{trigger} $\langle \textit{rbc-c2c}, \texttt{broadcast}, [\texttt{down}, M, c], u, c \rangle$ \Comment{\textit{see Algorithm \ref{alg:bc_committees}}}	
			}
			\Else{
				\textbf{trigger} $\langle \textit{rbc-c2n}, \texttt{broadcast}, [\texttt{deliver}, M], u, c \rangle$ \Comment{\textit{see Algorithm \ref{alg:bc_committee_node}}}
			}
		}	
	}
	
	\vspace*{0.3\baselineskip}		
	
	{\textit{\textbf{Wrapper Events:}}}
	
	\vspace*{0.2\baselineskip}		
	
	\OnEvent{$\langle \text{\textit{rbc-n2c}}, \texttt{deliver}, [\texttt{up}, M, u] \rangle$ \Or $\langle \text{\textit{rbc-c2c}}, \texttt{deliver}, [\texttt{up}, M, u] \rangle $}{
		\textbf{trigger} $\langle \textit{rbc}, \texttt{cast-up}, M, u \rangle$
	}

	\vspace*{0.3\baselineskip}		
	
	\OnEvent{$\langle \text{\textit{rbc-c2c}}, \texttt{deliver}, [\texttt{down}, M, u] \rangle $}{
		\textbf{trigger} $\langle \textit{rbc}, \texttt{cast-down}, M, u \rangle$
	}	
	\vspace*{0.3\baselineskip}		
	
	\OnEvent{$\langle \text{\textit{rbc-c2n}}, \texttt{deliver}, [\texttt{deliver}, M] \rangle $}{
		\textbf{trigger} $\langle \text{\textit{rbc}}, \texttt{deliver}, M \rangle $
	}
	
\end{algorithm}

\begin{theorem}
	\label{thm:reliable_broadcast}
	Given a system of witness committees as in Definition \ref{def:witness_committees} with with availability $\alpha$ and committee size parameter $\beta$. Algorithm \ref{alg:bc} solves the reliable broadcast problem from Def.\ \ref{def:reliable_broadcast} with $\bigO(\delta \cdot \beta)$ total sent messages per node for any $\delta$ with \smash{$ \tfrac{1}{\alpha} + 1 \leq \delta \leq n$}. The algorithm works even in the asynchronous setting. In partially synchronous setting it takes $\bigO(\log_\delta n)$ rounds until completion after global stabilization.
\end{theorem}

\begin{proof}
	
	\textbf{Integrity:} Since Algorithm \ref{alg:bc} uses Algorithms \ref{alg:bc_node_committee}, \ref{alg:bc_committee_node} and \ref{alg:bc_committees} as subroutines, its integrity is guaranteed by the integrity of these subroutines as shown in Lemmas \ref{lem:bc_node_committee}, \ref{lem:bc_committee_node} an \ref{lem:bc_committees}.
	
	\textbf{Validity and Agreement:} 
	For validity and agreement we rely on the structure of the broadcast tree $T$ (see Definition \ref{def:broadcast_tree} and Lemma \ref{lem:broadcast_tree}) and the corresponding properties of the subroutines Algorithm \ref{alg:bc_node_committee}, \ref{alg:bc_committee_node} and \ref{alg:bc_committees} due to Lemmas \ref{lem:bc_node_committee}, \ref{lem:bc_committee_node} an \ref{lem:bc_committees}. Algorithm \ref{alg:bc} sends the message from the leaf in $T$ that corresponds to the sending node, up to the root committee. From the root, this message is cast down the whole tree $T$ to all leaves. Since each node in $V$ is a leaf in $T$ and because all used subroutines satisfy the validity and agreement condition, each node will deliver the same message by the sender, which will happen eventually, given that the sender is honest.

	\textbf{Performance:} By Lemma \ref{lem:broadcast_tree}, we can choose as the degree $\delta$ of $T$ any value that is at least $\tfrac{1}{\alpha} +1$, where $\alpha$ defines the number of valid committees available to form inner nodes of $T$. Any such $T$ with degree $\delta$ can be computed locally by each node due to the common knowledge of the system of witness committees. The height $\bigO(\log_\delta n)$ of $T$ implies that the broadcast is done after $\bigO(\log_\delta n)$ communication steps between two layers of the tree. 
	
	Any communication step between two layers in the tree incurs at most $\bigO(\delta \beta)$ sent messages per node, because at most $\bigO(\beta)$ messages per node is incurred by the subroutines as shown  in Lemmas \ref{lem:bc_node_committee}, \ref{lem:bc_committee_node} and \ref{lem:bc_committees}, due to the degree $\delta$ of $T$ and since each node is part of only a constant number of committees (Definition \ref{def:witness_committees}). Finally, any committee and individual node is involved in at most two such communication steps (once for casting up and once for casting down), therefore $\bigO(\delta \beta)$ is a bound for the total number of messages per node.
\end{proof}

\subsection{Reliable Aggregation}

Our basic broadcast routines among witness committees and the broadcast tree can also be used to implement an algorithm that solves what we call the reliable aggregation problem. In this problem, a subset of nodes has an input and we want to compute the result of an aggregation function of all inputs. 

Roughly, an aggregation function allows to break the aggregation of multiple inputs down into aggregating subsets of inputs, and then aggregating the intermediate results, which is the case for the maximum, sum or xor functions to give some examples. Aggregation functions are defined differently in the literature, but in this work we use the following

\begin{definition}[Aggregation Function]
	\label{def:aggregation_function}
	A function $f : D^k \to I$ for any $k \in \mathbb N$ is an aggregation function if it has arbitrary arity $k$ (i.e., accepts any number of arguments from $D$) and the following properties. For any values $x_i$ in the domain $D$, the function $f$ satisfies associativity $f(x_1,x_2,x_3) = f(f(x_1,x_2),x_3)) = f(x_1,f(x_2,x_3))$ and symmetry $f(x_1,x_2, \ldots, x_k) =  f(x_{\pi(1)},x_{\pi(2)}, \ldots, x_{\pi(k)})$ for any permutation $\pi$ over $\{1, \dots, k\}$ and any $k \in \mathbb N$.
\end{definition}

We cast this problem into the Byzantine setting, where we would like every node to output the same value, but at least the inputs of the honest nodes need to be included into the final result. The problem is given formally as follows.

\begin{definition}[Reliable Aggregation]
	\label{def:reliable_aggregation}
	Let $f$ be an aggregation function.
	Let $V_{agg} \subseteq V_h$ be a subset of honest nodes, where each $v \in V_{agg}$ has an initial value $x_v$. %
	The reliable aggregation problem is solved as soon as the following holds.\footnote{We assume that any value in the image of the aggregation function fits into a generic message. The communication load measured in actual bits increases proportionally in the bit size of aggregated values.}
	\begin{itemize}
		\item \textbf{Validity}: Each $v \in V_h$ eventually outputs $f\big((x_v)_{v \in V'}\big)$ for some $V' \subseteq V$ with $V_{agg} \subseteq V'$.
		\item \textbf{Agreement}: If some $v \in V_h$ outputs a value $y$, then every node in $V_h$ outputs $y$.
	\end{itemize}	
\end{definition}

As such a reliable aggregation would allow us to solve the consensus problem as given in Definition \ref{def:consensus} in the crash-fail setting via a simple reduction, we cannot hope solve it in the asynchronous model \cite{DBLP:journals/jacm/FischerLP85}. More specifically, we need to be able to distinguish the situation where a Byzantine node refuses to commit a value for inclusion into the aggregation, from one where the submission by an honest node is delayed. 

This justifies making additional assumptions. In particular, compared to reliable broadcast, we have to revert to a synchrony assumption at least for the part where nodes commit to aggregation values. Further, we rely again on the pre-computed system of witness committees. The idea is to force each node that wants to include its value in the aggregation to commit its value to a dedicated witness committee within a deadline.

The algorithm works as follows. Each node commits its aggregation value to its parent committee in $T$. After a member of a committee core $C_u$ receives an aggregation value for the first time from a child in $T$, it waits until all submitted values $x_1, \ldots, x_k$ (where $k \leq \delta$) have been delivered, assuming that this will happen within the deadline. 
The delivered values are aggregated and the result submitted to the parent committee. From there on the process is repeated until the root committee aggregates the final result and broadcast it back over the tree to every node (re-using a subroutine from Algorithm \ref{alg:bc}).

\begin{algorithm}
	\caption{\textbf{Reliable Aggregation} \Comment{\textit{executed by $v \in V$}}}
	\label{alg:rag}
	
	\SetKwInOut{Input}{precondition}
	\SetKwInOut{Output}{output}	
	\SetKwInOut{Implements}{implements}
	\SetKwInOut{Requires}{requires}	
	\SetKwProg{On}{upon}{ do}{end}
	\SetKwProg{OnEvent}{upon event}{ do}{end}
	\SetKw{Or}{or}
	\SetKw{And}{and}
	
	\DontPrintSemicolon		
	
	\Implements{Reliable aggregation instance \textit{rag} that solves the reliable aggregation problem according to Definition \ref{def:reliable_aggregation}, where $V_{agg}$ are honest nodes that trigger $\langle \textit{rag}, \texttt{aggregate}, f, x \rangle$ in the same round.}
	\Input{System of witness committees satisfying Definition \ref{def:witness_committees}. Broadcast tree $T$ with degree $\delta$ as described in Definition \ref{def:broadcast_tree}.} %
	
	\vspace*{0.3\baselineskip}	
	
	\OnEvent{$\langle \textit{rag}, \texttt{aggregate}, f, x \rangle$}{
		$u \gets \text{ parent of } v \text{ in } T$
		
		\textbf{trigger} $\langle \textit{rbc-n2c}, \texttt{broadcast}, [\texttt{submit}, f, x, u], u \rangle$ \Comment{\textit{see Algorithm \ref{alg:bc_node_committee}}}
		
	}
	
	\vspace*{0.3\baselineskip}		
	
	\On{\normalfont triggering $\langle \textit{rag}, \texttt{aggregate-up}, f, x, u \rangle$ for the first time}{
		wait until message deadline expires \Comment{\textit{guarantees inclusion of values from $V_{agg}$}}
				
		$x_1, \dots, x_k \gets $ all values obtained via event $\langle \textit{rag}, \texttt{aggregate-up}, f, x_i, u \rangle$
		
		\If{$u \neq r$}{
			
			$p \gets \text{ parent of } u \text{ in } T$
			
			\textbf{trigger} $\langle \textit{rbc-c2c}, \texttt{broadcast}, [\texttt{submit}, f, f(x_1, \ldots, x_k), p], u, p \rangle$ \Comment{\textit{see Algorithm \ref{alg:bc_committees}}}	
		}
		
		\If{$u = r$}{			
			\textbf{trigger} $\langle \textit{rbc}, \texttt{cast-down}, [f, f(x_1, \ldots, x_k)], u \rangle$ \Comment{\textit{reuse \text{cast-down} event of Alg.\ \ref{alg:bc}}}
		}
		
	}

	\vspace*{0.3\baselineskip}		

	{\textit{\textbf{Wrapper Events:}}}
	
	\vspace*{0.2\baselineskip}				
	
	\OnEvent{$\langle \textit{rbc-n2c}, \texttt{deliver}, [\texttt{submit},f,x,u] \rangle$}
	{
		\textbf{trigger} $\langle \textit{rag}, \texttt{aggregate-up}, f, x, u \rangle$
	}
	
	\vspace*{0.3\baselineskip}		
	
		\OnEvent{$\langle \textit{rbc-c2c}, \texttt{deliver}, [\texttt{submit},f,x,u] \rangle$}
	{
		\textbf{trigger} $\langle \textit{rag}, \texttt{aggregate-up}, f, x, u \rangle$
	}
	
	\vspace*{0.3\baselineskip}		
	
	\OnEvent{$\langle \textit{rbc}, \texttt{deliver}, y \rangle$}{
		\textbf{trigger} $\langle \textit{rag}, \texttt{output}, y \rangle$
	}
	
\end{algorithm}

It remains to prove the Algorithm \ref{alg:rag} solves the reliable aggregation problem.

\begin{theorem}
	\label{thm:reliable_aggregation}
	Given a system of witness committees as in Definition \ref{def:witness_committees} with availability $\alpha$ and committee size parameter $\beta$. Algorithm \ref{alg:bc} solves the reliable aggregation problem defined in Def.\ \ref{def:reliable_aggregation} in the synchronous setting with $\bigO(\delta \cdot \beta)$ total sent messages per node in $\bigO(\log_\delta n)$ rounds, for any $\delta$ with \smash{$ \tfrac{1}{\alpha} + 1 \leq \delta \leq n$}.
\end{theorem}

\begin{proof}
	\textbf{Validity:} By algorithm design and assuming synchrony, any node $v \in V_{agg}$ submits its value to the parent committee in $T$ in the first round. Each member of the common core of that committee will deliver all committed values and can aggregate them locally using $f$. This process is then repeated until the root committee has aggregated all values. As any such transmission is based on reliable broadcasts (Lemmas \ref{lem:bc_node_committee}, \ref{lem:bc_committees}) the members of the common core of the root committee will aggregate the same value. This value must include the values of each $v \in V_{agg}$.
	
	\textbf{Agreement:} After the root committee computed the aggregation result, it will start the subroutine to broadcast it down the tree to all nodes from the reliable broadcast algorithm \ref{alg:bc}, thus all nodes will output the same value due to the same argument as in the proof of Theorem \ref{thm:reliable_broadcast}.
	
	\textbf{Performance:} This is also largely analogous to the proof of Theorem \ref{thm:reliable_broadcast}, with the procedure being completed after $\bigO(\log_\delta n)$ communication steps between two layers in the tree. The main difference is that here, each node and committee transmits one aggregation value up in $T$ whereas in the reliable broadcast Algorithm \ref{alg:bc} only the committees in the branch of $T$ of the sender had to send a value up. Asymptotically, this is at most as expensive as broadcasting a value down in $T$ from the root committee to all nodes which we do in both algorithms, thus the asymptotic number of messages each node sends per round is the same (i.e., $\bigO(\delta \beta)$).
\end{proof}

\subsection{Consensus and Common Coins via Reliable Aggregation}

In this subsection we give applications for our reliable aggregation routine. A key application in the synchronous Byzantine setting is computing a common coin beacon. Each node generates a local random coin, and the resulting bits are aggregated via an XOR operation. We assume a 1-late adversary, which must fix the local coins for its compromised nodes based on all information up to the previous round, i.e., before honest nodes finalize their own random coins. Consequently, as long as at least one honest node’s coin is included, the adversary cannot affect the distribution of the aggregated coin.

\begin{theorem}
	\label{thm:common_coin}
	Given a system of witness committees as in Definition \ref{def:witness_committees} with availability $\alpha$ and committee size parameter $\beta$. Assume that the network is synchronous and the adversary is 1-late. There is an algorithm such that each honest node outputs a sequence of $k$ random bits each distributed with probability $\tfrac{1}{2}$. The algorithm sends $\bigO(\delta \cdot k \cdot \beta)$ total bits per node and takes $\bigO(\log_\delta n)$ rounds, for any $\delta$ with \smash{$ \tfrac{1}{\alpha} + 1 \leq \delta \leq n$}.
\end{theorem}

\begin{proof}
	The algorithm is a simple application of reliable aggregation, where each node first determines a sequence of $k$ local random coins each distributed with Ber$(\tfrac{1}{2})$, which will be its input for the aggregation. The function \smash{$f := \bigoplus^k : \big(\{0,1\}^k\big)^* \to \{0,1\}^k$} takes any number of $k$-bit strings as input and returns a $k$-bit string, where the $i$-th bit corresponds to the logic xor of all $i$-th bits of all input bit-strings. This is an aggregation function in the sense of Definition \ref{def:aggregation_function}. 
	
	There exists at least one honest node $v$ (due to witness committees having honest cores) and $v$ submits a random $k$ bit string to the aggregation of $f$ using Algorithm \ref{alg:rag} solving the aggregation problem in Definition \ref{def:reliable_aggregation}. Since the adversary is 1-late, the bit-strings submitted by Byzantine nodes are independent of those submitted by the honest nodes, thus the resulting aggregation value is completely random.
	The length $k$ of the desired random bit string increases the number of bits per node proportionally. All other properties are directly inherited from Theorem \ref{thm:reliable_aggregation}.
\end{proof}

As second application we solve the consensus problem in the synchronous setting, which is defined as follows.

\begin{definition}[Consensus]
	\label{def:consensus}
	Assume each node $v \in V$ initially proposes a value $x_v$ from a given set $\{x_1, \dots ,x_k\}$ of values.\footnote{We assume proposed values can be encoded with $\log k$ bits, else the workload per node increases proportionally.} The consensus problem is solved when the following holds.
	\begin{itemize}
		\item \textbf{Termination:} Each honest node eventually decides on a value.
		\item \textbf{Validity}: If all $v \in V_h$ propose the same value, then every node decides on that value.
		\item \textbf{Agreement}: Every $v \in V_h$ decides on the same value from the set $\{x_1, \dots ,x_k\}$.
	\end{itemize}	
\end{definition}

We aim for a deterministic solution that solves the consensus problem even for $t < \tfrac{n}{2}$ Byzantine nodes given that we have a system of witness committees. We remark that the actual Byzantine threshold and randomization requirements are masked by the existence of a  system of witness committees.

\begin{theorem}
	\label{thm:consensus}
	Given a system of witness committees as in Definition \ref{def:witness_committees} with availability $\alpha$ and committee size parameter $\beta$ and at most $t<\tfrac{n}{2}$ Byzantine nodes. The consensus problem  given in Definition \ref{def:consensus} with $k$ proposed values can be solved in the synchronous setting with $\bigO(\delta \beta k \cdot \log n)$ bits total communication per node in $\bigO(\log_\delta n)$ rounds, for any \smash{$\tfrac{1}{\alpha} \leq \delta \leq n$}.
\end{theorem}

\begin{proof}
	The algorithm works as follows. 	
	First, we make a count of the proposed values using the reliable aggregation subroutine, i.e., we return a set of pairs $(x_i,c_i)$ where $c_i$ is the number of nodes that committed to opinion $x_i$. Note that this is indeed an aggregation function and that Byzantine nodes are forced to commit to an opinion, so the counts returned by the reliable aggregation will be the same for everyone. Further, the result of the aggregation function has length at most $k \cdot \log n$ bits.

	If one opinion has a majority of at least $\tfrac{n}{2}$ in the count, then all nodes immediately decide this value, knowing that at least one honest node must have proposed it due to $t < \tfrac{n}{2}$, which ensures \textbf{validity}. 	
	Else, there exist at least two honest nodes that proposed different values and it does not matter which value we decide as long as the decision is consistent among honest nodes. We can for example decide on the value with the largest count breaking ties by lexicographic order of the values. This ensures \textbf{agreement} and \textbf{termination}. 
	
	The \textbf{performance} properties are inherited directly from the according properties of reliable aggregation (see Theorem \ref{thm:reliable_aggregation}).
\end{proof}

Finally, we remark that the adversary could use any deterministic rule to select the default value to influence the decided value in case there is no absolute majority. Alternatively, we can randomize the default value as follows. We use Theorem \ref{thm:common_coin} to determine a common random string of length $\bigO(\log k)$. The default value can then be determined fairly by the result of this random bit-string. The common randomness cannot be influenced by the adversary, but it still ensures consensus on the default value. Furthermore, this has no asymptotic disadvantages to running time and message complexity.

\newpage

\appendix

\section{Probabilistic Concepts}
\label{sec:probabilistic_concepts}

We give a few basic definitions and principles for the probabilistic properties of the protocols introduced in this paper.

\begin{definition}[Negligible Function]
	\label{def:negligible_function}
	A function $f$ is negligible if for any polynomial $\pi$ there is a constant $\lambda_0 \geq 0$, s.t., for any $\lambda \geq \lambda_0$ it is $f(\lambda) \leq (\pi (\lambda))^{-1}$.
\end{definition}

\begin{remark}
	We often use that for any constant $c > 0$, the function $f(\lambda) = e^{-c\cdot \lambda}$ is negligible w.r.t.\ $\lambda$. 
\end{remark}

\begin{definition}[All But Negligible Probability]
	\label{def:all_but_neglible_probability}
	An event is said to occur with all but negligible probability with respect to some parameter $\lambda$ if the probability of the event not happening is a function in $\lambda$ that is negligible w.r.t. $\lambda$.
\end{definition}

In the literature, randomized distributed algorithms are often shown to be successful \emph{with high probability}, which expresses the probability of failure as a function that decreases inversely with the input size $n$ of the problem, i.e., $n^{-c}$ for some fixed $c>0$. This is often quite convenient as it does no require to consider the failure threshold as a parameter in the analysis.

Besides depending on the network size $n$ instead of some external security parameter, the main difference to the notion of a negligible function compared to events that do not occur with probability $n^{-c}$ is that $c >0$ is fixed ``in advance'', i.e., the probability of failure does not have to be smaller than \textit{every} polynomial (in the sense of Definition \ref{def:negligible_function}) but only smaller than some fixed polynomial $n^{-c}$.

\begin{definition}[With High Probability]
	\label{def:whp}
	An event is said to hold with high probability (w.h.p.), if for any fixed constant $c\geq1$ that the event occurs with probability at least $1-n^{-c}$  for sufficiently large $n$.
\end{definition}

The disadvantage of the notion w.h.p.\ is that it usually looks at the asymptotic behavior of a system (i.e., for large $n$), and does not necessarily guarantee a high safety level level for small $n$, which can be a requirement in practice. We hold ourselves to a somewhat higher standard, requiring our protocols to succeed with high probability \textit{and} all but negligible probability with respect to $\lambda$, thus also guaranteeing a high level of safety even for small networks, which we call \textit{with high confidence} (w.h.c.), see Definition \ref{def:whc}. The following lemmas show that we can apply the union bound to show that if a polynomial (in $n$ or $\lambda$) number of events each occur w.h.c., then all of them occur w.h.c.

\begin{lemma}[Union Bound for Events Holding with all but Negligible Probability]
	\label{lem:unionbound_negl}
	Let $E_1, \ldots ,E_k$ be events each of which occurs with all but negligible probability with respect to security parameter $\lambda$ (see Definition \ref{def:all_but_neglible_probability}). Further, assume $k$ is polynomially bounded in $\lambda$, i.e., $k \leq \pi(\lambda)$ for some fixed polynomial $\pi$.  Then $E := \bigcap_{i=1}^{k} E_i$ occurs with all but negligible probability.
\end{lemma}

\begin{proof}
	Let $\rho$ be an arbitrary polynomial. Let $\sigma(\lambda) := \rho(\lambda)\pi(\lambda)$. As $E_1, \dots , E_k$ satisfy Definition \ref{def:all_but_neglible_probability} we obtain the following. There exist $\lambda_1, \ldots , \lambda_k \in \mathbb{N}$ such that for all $i \in \{1, \ldots, k\}$ we have $\mathbb{P}(\overline{E_i}) \leq \frac{1}{\sigma(\lambda)}$ for any $\lambda > \lambda_i$. With Boole's inequality (a.k.a.\ ``Union Bound'') we obtain
	\begin{align*}
		\mathbb{P}\big(\overline{E}\big) 
		= \mathbb{P}\Big(\bigcup_{i=1}^{k} \overline{E_i} \Big) 
		\leq \sum_{i=1}^{k} \mathbb{P}(\overline{E_i}) 
		\leq \sum_{i=1}^{k} \frac{1}{\sigma(\lambda)} 
		\leq  \frac{\pi(\lambda)}{\sigma(\lambda)} 
		= \frac{\pi(\lambda)}{\rho(\lambda)\pi(\lambda)} 
		= \frac{1}{\rho(\lambda)}
	\end{align*}
	for any $\lambda \geq \lambda_0 := \max(\lambda_1, \ldots ,\lambda_k)$. 
\end{proof}

\begin{lemma}[Union Bound for Events Holding w.h.p.]
	\label{lem:unionbound_whp}
	Let $E_1, \ldots ,E_k$ be events each of which occurs w.h.p.\ (see Definition \ref{def:whp}). Further, assume $k$ is polynomially bounded in $n$, i.e., $k \leq n^d$ for all $n \geq n_0$ for some monomial with fixed degree $d>0$.  Then $E := \bigcap_{i=1}^{k} E_i$ occurs w.h.p.
\end{lemma}

\begin{proof}
	Let $c>0$ be arbitrary but fixed. As $E_1, \dots , E_k$ satisfy Definition \ref{def:whp} we obtain the following. Let $c':= c+d$, there exist $n_1, \ldots , n_k \in \mathbb{N}$ such that for all $i \in \{1, \ldots, k\}$ we have $\mathbb{P}(\overline{E_i}) \leq \frac{1}{n^{c'}}$ for any $n > n_i$. With Boole's inequality we obtain
	\begin{align*}
		\mathbb{P}\big(\overline{E}\big) = 
		\mathbb{P}\Big(\bigcup_{i=1}^{k} \overline{E_i} \Big) 
		\leq \sum_{i=1}^{k} \mathbb{P}(\overline{E_i}) 
		\leq \sum_{i=1}^{k} \frac{1}{n^{c'}} 
		\leq  n^{d - c'}
		= n^{-c}
	\end{align*}
	for any $n \geq \max(n_0, \ldots ,n_k)$. 
\end{proof}

We combine the two Lemmas \ref{lem:unionbound_negl} and \ref{lem:unionbound_whp} to obtain a union bound for events holding with high confidence.

\begin{corollary}[Union Bound for Events Holding w.h.c.]
	\label{cor:unionbound_whc}
	 Let $E_1, \ldots ,E_k$ be events each of which occurs w.h.c.\ with respect to security parameter $\lambda$ (see Definition \ref{def:whc}). Further, assume $k$ is polynomially bounded in $\lambda$ and $n$, i.e., $k \leq \pi(\lambda)$ and $k \leq \pi(n)$ for some fixed polynomial $\pi$.  Then $E := \bigcap_{i=1}^{k} E_i$ occurs w.h.c.
\end{corollary}

We use the following Chernoff bounds for Bernoulli variables in some of our proofs.

\begin{lemma}[Chernoff Bound]
	\label{lem:chernoffbound}
	Let $X = \sum_{i=1}^n X_i$ for i.i.d.\ random variables $X_i \in \{0,1\}$ and $\mathbb{E}(X) \leq \mu_H$ and $\delta \geq 0$, then
	$$\mathbb{P}\big(X \geq (1 \!+\! \delta) \mu_H\big) \leq \exp\Big(\!-\!\frac{\min(\delta,\delta^2)\mu_H}{3}\Big),$$
	Similarly, for $\mathbb{E}(X) \geq \mu_L$ and $0 \leq \delta \leq 1$ we have
	$$\mathbb{P}\big(X \leq (1 \!-\! \delta) \mu_L\big) \leq \exp\Big(\!-\!\frac{\delta^2\mu_L}{2}\Big).$$
\end{lemma}

\newpage

\section{Glossary of Commonly Used Variables}
\label{sec:variable_glossary}

The following table gives an overview of the variables that we use in our algorithms.

\renewcommand{\arraystretch}{1.3}
\begin{table}[H]
	\label{tab:parameters}
	\begin{tabular}{@{}lp{0.75\linewidth}@{}}
		\toprule
		\textbf{Parameter} & \textbf{Description} \\
		\midrule
		$A_{uv}, B_{uv}, C_{uv}$  
		& (Preliminary) witness committees during Phases~A, B, and C that describe $v$'s view of $u$'s committee. $C_{uv}$ denotes the final committees and the output of our algorithm. \\
		
		$C_u$ 
		& Common core of honest nodes in $u$'s committee (see Lemma \ref{lem:common_core}). \\
		
		$M_v$ & Set of nodes $u \in V$ that node $v$ samples and whose committees it intends to join.\\
		
		$N_{uv}$ & Set of nodes $w \in V$ from which $v$ requests $B_{uw}$ to construct $C_{uv}$.\\
		
		$\alpha$ 
		& Availability parameter for witness committees; fraction of nodes $u \in V$ that have nonempty committees $C_{uv}$ (see Definition \ref{def:witness_committees}).\\
		
		$\beta$ 
		& Committee size parameter (see Definition \ref{def:witness_committees} and Lemma \ref{lem:size_committees}). \\
		
		$\gamma$ 
		& Sampling parameter for creating the preliminary committees $A_{uv}$ in Phase~A (see Lemma \ref{lem:size_committees}). \\
		
		$\delta$ 
		& Degree parameter of the broadcast tree (see Definition \ref{def:broadcast_tree} and Lemma \ref{lem:broadcast_tree}).\\
		
		$\zeta$ 
		& Sampling parameter for assembling the final committees $C_{uv}$ in Phase~C (see Lemmas \ref{lem:respond}, \ref{lem:prelim_committees_c} and \ref{lem:min_valid_prelim_committees_c}). \\
		
		$\lambda$ 
		& Security parameter ensuring negligible failure probability even for small networks (e.g., $e^{-\lambda}$, see Definitions \ref{def:all_but_neglible_probability} and \ref{def:whc}). \\
		
		$\sigma$ 
		& Bandwidth per node (number of bits a node can send per round).\\
		
		\bottomrule
	\end{tabular}
\end{table}

\newpage

\printbibliography

\end{document}